\RequirePackage{amsmath}
\documentclass[runningheads]{llncs}

\usepackage[T1]{fontenc}

\usepackage{amsmath}
\usepackage{amssymb}
\usepackage{graphicx}
\usepackage{enumerate}
\usepackage{color}

\newtheorem{thm}[theorem]{Theorem}
\newtheorem{lem}[lemma]{Lemma}
\newtheorem{cor}[lemma]{Corollary}
\newtheorem{pro}[property]{Property}

\newcounter{assumption}
\newtheorem{ass}[assumption]{Assumption}

\newcommand{\Xomit}[1]{#1}

\newcommand{\Plane}[0]{\mathbb{R}^2}

\newcommand{\MED}[2]{$(#1|#2)$-medianoid}
\newcommand{\CEN}[2]{$(#1|#2)$-centroid}
\newcommand{\MEDR}[2]{$(#1|#2)_R$-medianoid}
\newcommand{\CENR}[2]{$(#1|#2)_R$-centroid}

\newcommand{\CR}[1]{{C_R({#1})}}
\newcommand{\rr}[0]{{\gamma}}
\newcommand{\Cr}[1]{{C_\rr({#1})}}
\newcommand{\W}[2]{W({#2}|{#1})}
\newcommand{\B}[2]{B({#2}|{#1})}
\newcommand{\F}[2]{F({#2}|{#1})}
\newcommand{\y}[3][]{y_{#1}({#3}|{#2})}

\newcommand{\Tr}[3][]{T^r_{#1}({#3}|{#2})}
\newcommand{\Tl}[3][]{T^l_{#1}({#3}|{#2})}

\newcommand{\tr}[3][]{t^r_{#1}({#3}|{#2})}
\newcommand{\tl}[3][]{t^l_{#1}({#3}|{#2})}

\title{The $(1|1)$-Centroid Problem on the Plane Concerning Distance Constraints
\thanks{Research supported by
under Grants No. MOST 103-2221-E-005-042, 103-2221-E-005-043.
}
}

\titlerunning{The $(1|1)_R$-Centroid Problem on the Plane}

\author{Hung-I Yu \and Tien-Ching Lin \and D. T. Lee}
\institute{Institute of Information Science, Academia Sinica,\\
Nankang, Taipei 115, Taiwan,\\
\email{\{herbert,kero,dtlee\}@iis.sinica.edu.tw}}

\begin{document}
\maketitle
\begin{abstract}
In 1982, Drezner proposed the \CEN{1}{1} problem on the plane,
in which two players, called the leader and the follower,
open facilities to provide service to customers in a competitive manner.
The leader opens the first facility, and the follower opens the second.
Each customer will patronize the facility closest to him
(ties broken in favor of the first one),
thereby decides the market share of the two facilities.
The goal is to find the best position for the leader's facility
so that its market share is maximized.
The best algorithm of this problem is 
an $O(n^2 \log n)$-time parametric search approach,
which searches over the space of market share values.

In the same paper, Drezner also proposed a general version 
of \CEN{1}{1} problem by introducing a minimal distance constraint $R$,
such that the follower's facility is not allowed to be
located within a distance $R$ from the leader's.
He proposed an $O(n^5 \log n)$-time algorithm for this general version 
by identifying $O(n^4)$ points as the candidates of the optimal solution
and checking the market share for each of them.
In this paper, we develop a new parametric search approach 
searching over the $O(n^4)$ candidate points,
and present an $O(n^2 \log n)$-time algorithm for the general version, 
thereby close the $O(n^3)$ gap between the two bounds.

\end{abstract}

\section{Introduction} \label{introduction}

In 1929, economist Hotelling introduced the first competitive location
problem in his seminal paper \cite{Hotelling-29}.
Since then, the subject of competitive facility location
has been extensively studied by researchers in the fields of
spatial economics, social and political sciences, and operations research,
and spawned hundreds of contributions in the literature.
The interested reader is referred to the following survey papers
\cite{Dasci-11,Eiselt-97,Eiselt-93,Eiselt-15,%
Hakimi-90,Hansen-90,Plastria-01,Santos-07}.

Hakimi \cite{Hakimi-83} and Drezner \cite{Drezner-82} individually proposed
a series of competitive location problems in a leader-follower framework.
The framework is briefly described as follows.
There are $n$ customers in the market,
and each is endowed with a certain buying power.
Two players, called the \emph{leader} and the \emph{follower},
sequentially open facilities to attract the buying power of customers.
At first, the leader opens his $p$ facilities,
and then the follower opens another $r$ facilities.
Each customer will patronize the closest facility with all buying power
(ties broken in favor of the leader's ones),
thereby decides the market share of the two players.
Since both players ask for market share maximization,
two competitive facility location problems are defined under this framework.
Given that the leader locates his $p$ facilities at the set $X_p$ of $p$ points,
the follower wants to locate his $r$ facilities
in order to attract the most buying power,
which is called the \emph{\MED{r}{X_p}} problem.
On the other hand, knowing that the follower will
react with maximization strategy, 
the leader wants to locate his $p$ facilities
in order to retain the most buying power against the competition,
which is called the \emph{\CEN{r}{p}} problem.


Drezner \cite{Drezner-82} first proposed to study the two
competitive facility location problems on the Euclidean plane.
Since then, many related results \cite{Davydov-13,Drezner-82,Drezner-92,Hakimi-90,Lee-86}
have been obtained for different values of $r$ and $p$.
Due to page limit, here we introduce only
previous results about the case $r=p=1$.
For the \MED{1}{X_1} problem, Drezner \cite{Drezner-82} showed that
there exists an optimal solution arbitrarily close to $X_1$,
and solved the problem in $O(n \log n)$ time by sweeping technique.
Later, Lee and Wu \cite{Lee-86} obtained an $\mathrm{\Omega}(n \log n)$
lower bound for the \MED{1}{X_1} problem,
and thus proved the optimality of the above result.
For the \CEN{1}{1} problem, Drezner \cite{Drezner-82}
developed a parametric search based approach
that searches over the space of $O(n^2)$ possible market share values,
along with an $O(n^4)$-time test procedure
constructing and solving a linear program of $O(n^2)$ constraints,
thereby gave an $O(n^4 \log n)$-time algorithm.
Then, by improving the test procedure via
Megiddo's result \cite{Megiddo-83-2} for solving linear programs,
Hakimi \cite{Hakimi-90} reduced the time complexity to $O(n^2 \log n)$.

In \cite{Drezner-82}, Drezner also proposed a more general setting
for the leader-follower framework by introducing a \emph{minimal distance
constraint} $R \ge 0$ into the \MED{1}{X_1} problem and the \CEN{1}{1} problem,
such that the follower's facility is not allowed to be
located within a distance $R$ from the leader's.
The augmented problems are respectively called
the \emph{\MEDR{1}{X_1}} problem and \emph{\CENR{1}{1}} problem in this paper.
Drezner showed that the \MEDR{1}{X_1} problem
can also be solved in $O(n \log n)$ time
by using nearly the same proof and technique as for the \MED{1}{X_1} problem.
However, for the \CENR{1}{1} problem, he argued that
it is hard to generalize the approach for the \CEN{1}{1} problem
to solve this general version, due to the change of problem properties.
Then, he gave an $O(n^5 \log n)$-time algorithm
by identifying $O(n^4)$ candidate points on the plane,
which contain at least one optimal solution,
and performing medianoid computation on each of them.
So far, the $O(n^3)$ bound gap between
the two centroid problems remains unclosed.

In this paper, we propose an $O(n^2 \log n)$-time algorithm
for the \CENR{1}{1} problem on the Euclidean plane,
thereby close the gap last for decades.
Instead of searching over market share values,
we develop a new approach based on the parametric search technique
by searching over the $O(n^4)$ candidate points
mentioned in \cite{Drezner-82}.
This is made possible by making a critical observation
on the distribution of optimal solutions
for the \MEDR{1}{X_1} problem given $X_1$,
which provides us a useful tool to prune
candidate points with respect to $X_1$.
We then extend the usage of this tool
to design a key procedure to prune candidates
with respect to a given vertical line.

The rest of this paper is organized as follows.
Section \ref{prelinimary} gives formal problem definitions
and describes previous results in \cite{Drezner-82,Hakimi-90}.
In Section \ref{sec_Alg_line},
we make the observation on the \MEDR{1}{X_1} problem,
and make use of it to find a ``local'' centroid on a given line.
This result is then extended as a new pruning procedure
with respect to any given line in Section \ref{sec_Alg_plane},
and utilized in our parametric search approach for the \CENR{1}{1} problem.
Finally, in Section \ref{conclusion}, we give some concluding remarks.

\section{Notations and Preliminary Results} \label{prelinimary}


Let $V = \{v_1, v_2, \cdots, v_n\}$ be a set of
$n$ points on the Euclidean plane $\Plane$,
as the representatives of the $n$ customers.
Each point $v_i \in V$ is assigned with
a positive weight $w(v_i)$, representing its buying power.
To simplify the algorithm description,
we assume that the points in $V$ are in general position,
that is, no three points are collinear and
no two points share a common x or y-coordinate.

Let $d(u, w)$ denote the Euclidean distance between any two points $u,w \in \Plane$.
For any set $Z$ of points on the plane,
we define $W(Z) = \sum \{w(v)|v \in V \bigcap Z\}$.
Suppose that the leader has located his facility at $X_1 = \{x\}$,
which is shortened as $x$ for simplicity.
Due to the minimal distance constraint $R$ mentioned in \cite{Drezner-82},
any point $y' \in \Plane$ with $d(y', x) < R$ is
infeasible to be the follower's choice.
If the follower locates his facility at some feasible point $y$,
the set of customers patronizing $y$ instead of $x$ is defined as
$V(y|x) = \{v \in V|d(v, y) < d(v, x)\}$,
with their total buying power $\W{x}{y} = W(V(y|x))$.
Then, the largest market share that the follower can capture is denoted by the function
\begin{equation*} \label{eq1}
    W^*(x) = \max_{y \in \Plane, d(y,x) \ge R} \W{x}{y},
\end{equation*}
which is called the \emph{weight loss} of $x$.
Given a point $x \in \Plane$, the \MEDR{1}{x} problem is to find a \emph{\MEDR{1}{x}},
which denotes a feasible point $y^* \in \Plane$ maximizing the weight loss of $x$.

In contrast, the leader tries to minimize the weight loss of his own facility
by finding a point $x^* \in \Plane$ such that
\begin{equation*} \label{eq2}
    W^*(x^*) \le W^*(x)
\end{equation*}
for any point $x \in \Plane$.
The \CENR{1}{1} problem is to find a \emph{\CENR{1}{1}},
which denotes a point $x^*$ minimizing its weight loss.
Note that, when $R=0$, the two problems degenerate
to the \MED{1}{x} and \CEN{1}{1} problems.


\subsection{Previous approaches} \label{ssec_11cen}

In this subsection, we briefly review previous results for the \MEDR{1}{x},
\CEN{1}{1}, and \CENR{1}{1} problems in \cite{Drezner-82,Hakimi-90},
so as to derive some basic properties essential to our approach.

Let $L$ be an arbitrary line, which partitions the Euclidean plane into two half-planes.
For any point $y \notin L$, we define $H(L, y)$ as the close half-plane including $y$,
and $H^-(L, y)$ as the open half-plane including $y$ (but not $L$).
For any two distinct points $x, y \in \Plane$,
let $\B{x}{y}$ denote the perpendicular bisector of the line segment from $x$ to $y$.

Given an arbitrary point $x \in \Plane$,
we first describe the algorithm for finding a \MEDR{1}{x} in \cite{Drezner-82}.
Let $y$ be an arbitrary point other than $x$,
and $y'$ be some point on the open line segment from $y$ to $x$.
We can see that $H^-(\B{x}{y},y) \subset H^-(\B{x}{y'},y')$,
which implies the fact that $\W{x}{y'} = W(H^-(\B{x}{y'},y'))
\ge W(H^-(\B{x}{y},y)) = \W{x}{y}$,
It shows that moving $y$ toward $x$ does not diminish its weight capture,
thereby follows the lemma.

\begin{lem} \cite{Drezner-82} \label{lem_onC}
    There exists a \MEDR{1}{x} in $\{y \;|\; y \in \Plane, d(x, y) = R\}$.
\end{lem}

For any point $z \in \Plane$, let $\CR{x}$ and $\Cr{x}$ be
the circles centered at $z$ with radii $R$ and $\rr = R/2$, respectively.
By Lemma \ref{lem_onC}, finding a \MEDR{1}{x} can be reduced
to searching a point $y$ on $\CR{x}$ maximizing $\W{x}{y}$.
Since the perpendicular bisector $\B{x}{y}$ of each point
$y$ on $\CR{x}$ is a tangent line to the circle $\Cr{x}$,
the searching of $y$ on $\CR{x}$ is equivalent to
finding a tangent line to $\Cr{x}$ that partitions the most weight from $x$.
The latter problem can be solved in $O(n \log n)$ time as follows.
For each $v \in V$ outside $\Cr{x}$, we calculate its two tangent lines to $\Cr{x}$.
Then, by sorting these tangent lines according to
the polar angles of their corresponding tangent points with respect to $x$,
we can use the angle sweeping technique to check how much weight they partition.

\begin{thm} \cite{Drezner-82} \label{thm_find_med}
    Given a point $x \in \Plane$,
    the \MEDR{1}{x} problem can be solved in $O(n \log n)$ time.
\end{thm}

Next, we describe the algorithm of the \CENR{1}{1} problem in \cite{Drezner-82}.
Let $S$ be a subset of $V$.
We define $\mathcal{C}(S)$ to be the set of all circles $\Cr{v}$, $v \in S$,
and $CH(\mathcal{C}(S))$ to be the convex hull of these circles.
It is easy to see the following.

\begin{lem} \cite{Drezner-82} \label{lem_point_roundch}
    Let $S$ be a subset of $V$.
    For any point $x \in \Plane$, $W^*(x) \ge W(S)$ if $x$ is outside $CH(\mathcal{C}(S))$.
\end{lem}

For any positive number $W_0$, let $I(W_0)$ be the intersection
of all convex hulls $CH(\mathcal{C}(S))$, where $S \subseteq V$ and $W(S) \ge W_0$.
We have the lemma below.

\begin{lem} \cite{Drezner-82} \label{lem_point_ge_W0}
    Let $W_0$ be a positive real number.
    For any point $x \in \Plane$, $W^*(x) < W_0$ if and only if $x \in I(W_0)$.
\end{lem}

\Xomit{
\begin{proof}
Consider first the case that $x \in I(W_0)$.
By definition, $x$ intersects with every 
$CH(\mathcal{C}(S))$ of subset $S \subseteq V$ with $W(S) \ge W_0$.
Let $S' \subseteq V$ be any of such subsets.
Since $x \in CH(\mathcal{C}(S'))$, for any point $y$ feasible to $x$,
there must exist a point $v \in S$
such that $v \notin H^-(\B{x}{y},y)$,
implying that no feasible point $y$ 
can acquire all buying power from customers of $S'$.
It follows that 
no feasible point $y$ can acquire buying power larger than or equal to $W_0$,
i.e., $W^*(x) < W_0$.

If $x \notin I(W_0)$, there must exist a subset $S \subseteq V$ with $W(S) \ge W_0$,
such that $x \notin CH(\mathcal{C}(S))$.
By Lemma \ref{lem_point_roundch}, $W^*(x) \ge W(S) \ge W_0$.
\qed
\end{proof}
}

Drezner \cite{Drezner-82} argued that the set of all \CENR{1}{1}s
is equivalent to some intersection $I(W_0)$ for smallest possible $W_0$.
We slightly strengthen his argument below.
Let $\mathcal{W} = \{\W{x}{y} \;|\; x,y \in \Plane, d(x,y) \ge R\}$.
The following lemma can be obtained.

%
\begin{lemma} \label{lem_opt_region}
    Let $W^*_0$ be the smallest number in $\mathcal{W}$
    such that $I(W^*_0)$ is not null.
    A point $x$ is a \CENR{1}{1} if and only if $x \in I(W^*_0)$.
\end{lemma}

\begin{proof}
Let $W_{OPT}$ be the weight loss of some \CENR{1}{1} $x^*$.
We first show that $I(W_0)$ is null for any $W_0 \le W_{OPT}$.
Suppose to the contrary that it is not null
and there exists a point $x'$ in $I(W_0)$.
By Lemma \ref{lem_point_ge_W0}, $W^*(x') < W_0 \le W_{OPT}$,
which contradicts the optimality of $x^*$.
Moreover, since $I(W^*_0)$ is not null, we have that $W^*_0 > W_{OPT}$.

We now show that a point $x$ is a \CENR{1}{1} if and only if $x \in I(W^*_0)$.
If $x$ is a \CENR{1}{1}, we have that $W^*(x) = W_{OPT} < W^*_0$.
By Lemma \ref{lem_point_ge_W0}, $x \in I(W^*_0)$.
On the other hand, if $x$ is not a \CENR{1}{1},
we have that $W^*(x) > W_{OPT}$.
Since by definition $W^*(x) \in \mathcal{W}$,
we can see that $W^*(x) \ge W^*_0$.
Thus, again by Lemma \ref{lem_point_ge_W0}, $x \notin I(W^*_0)$.
\qed
\end{proof}

Although it is hard to compute $I(W^*_0)$ itself,
we can find its vertices as solutions to the \CENR{1}{1} problem.
Let $\mathcal{T}$ be the set of outer tangent lines
of all pairs of circles in $\mathcal{C}(V)$.
For any subset $S \subseteq V$, the boundary of $CH(\mathcal{C}(S))$ is
formed by segments of lines in $\mathcal{T}$ and arcs of circles in $\mathcal{C}(V)$.
Since $I(W_0)$ is an intersection of such convex hulls,
its vertices must fall within the set of intersection points
between lines in $\mathcal{T}$, between circles in $\mathcal{C}(V)$,
and between one line in $\mathcal{T}$ and one circle in $\mathcal{C}(V)$.
Let $\mathcal{T} \times \mathcal{T}$, $\mathcal{C}(V) \times \mathcal{C}(V)$,
and $\mathcal{T} \times \mathcal{C}(V)$
denote the three sets of intersection points, respectively.
We have the lemma below.

\begin{lem} \cite{Drezner-82} \label{lem_TTCCTC}
    There exists a \CENR{1}{1} in $\mathcal{T} \times \mathcal{T}$,
    $\mathcal{C}(V) \times \mathcal{C}(V)$, and $\mathcal{T} \times \mathcal{C}(V)$.
\end{lem}

Obviously, there are at most $O(n^4)$ intersection points,
which can be viewed as the \emph{candidates} of being \CENR{1}{1}s.
Drezner thus gave an algorithm by evaluating the weight loss
of each candidate by Theorem \ref{thm_find_med}.

\begin{thm} \cite{Drezner-82} \label{thm_find_cen}
    The \CENR{1}{1} problem can be solved in $O(n^5 \log n)$ time.
\end{thm}

We remark that, when $R = 0$, $CH(\mathcal{C}(S))$
for any $S \subseteq V$ degenerates to a convex polygon,
so does $I(W_0)$ for any given $W_0$, if not null.
Drezner \cite{Drezner-82} proved that in this case
$I(W_0)$ is equivalent to the intersection of
all half-planes $H$ with $W(H) \ge W_0$.
Thus, whether $I(W_0)$ is null can be determined
by constructing and solving a linear program of $O(n^2)$ constraints,
which takes $O(n^2)$ time by Megiddo's result \cite{Megiddo-83-2}.
Since $|\mathcal{W}| = O(n^2)$, by Lemma \ref{lem_opt_region},
the \CEN{1}{1} problem can be solved in $O(n^2 \log n)$ time \cite{Hakimi-90},
by applying parametric search over $\mathcal{W}$ for $W^*_0$.
Unfortunately, it is hard to generalize this idea to the case $R > 0$,
motivating us to develop a different approach.

\section{Local $(1|1)_R$-Centroid within a Line} \label{sec_Alg_line}

In this section, we analyze the properties of
\MEDR{1}{x}s of a given point $x$ in Subsection \ref{ssec_Point},
and derive a procedure that prunes candidate points with respect to $x$.
Applying this procedure, we study a restricted version of
the \CENR{1}{1} problem in Subsection \ref{ssec_LineLocal},
in which the leader's choice is limited to a given non-horizontal line $L$,
and obtain an $O(n \log^2 n)$-time algorithm.
The algorithm is then extended as the basis of the test procedure
for the parametric search approach in Section \ref{sec_Alg_plane}.


\subsection{Pruning with Respect to a Point} \label{ssec_Point}

Given a point $x \in \Plane$ and
an angle $\theta$ between 0 and $2\pi$,
let $\y{x}{\theta}$ be the point on $\CR{x}$ with
polar angle $\theta$ with respect to $x$.%
\footnote{We assume that a polar angle is measured
counterclockwise from the positive x-axis.}
We define $MA(x) = \{\theta\:|\:\W{x}{\y{x}{\theta}}
= W^*(x), 0 \le \theta < 2\pi\}$,
that is, the set of angles $\theta$ maximizing
$\W{x}{\y{x}{\theta}}$ (see Figure \ref{fig_MA}).
It can be observed that, for any $\theta \in MA(x)$
and sufficiently small $\epsilon$,
both $\theta+\epsilon$ and $\theta-\epsilon$ belong to $MA(x)$,
because each $v \in V(\y{x}{\theta}|x)$ does not
intersect $\B{x}{\y{x}{\theta}}$ by definition.
This implies that angles in $MA(x)$ form
open angle interval(s) of non-zero length.

\begin{figure}[t]
    \centering
    \includegraphics[scale=1]{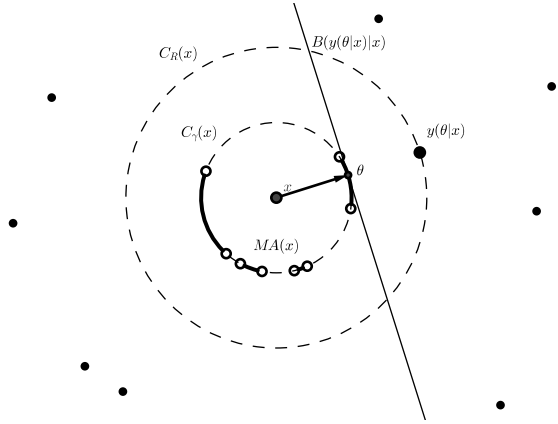}
    \caption{The black arcs represent the intervals of angles in $MA(x)$,
    whereas the open circles represent the open ends of these intervals.}
    \label{fig_MA}
\end{figure}

To simplify the terms, let $\W{x}{\theta} = \W{x}{\y{x}{\theta}}$
and $\B{x}{\theta} = \B{x}{\y{x}{\theta}}$
in the remaining of this section.
Also, let $\F{x}{\theta}$ be the line
passing through $x$ and parallel to $\B{x}{\theta}$.
The following lemma provides the basis for pruning.

\begin{lem} \label{lem_point_MAprune}
    Let $x \in \Plane$ be an arbitrary point,
    and $\theta$ be an angle in $MA(x)$.
    For any point $x' \notin
    H^-(\F{x}{\theta}, \y{x}{\theta})$,
    $W^*(x') \ge W^*(x)$.
\end{lem}

\Xomit{
\begin{proof}
Since $x' \notin H^-(\F{x}{\theta}, \y{x}{\theta})$
and $\y{x}{\theta} \in H^-(\F{x}{\theta}, \y{x}{\theta})$,
by the definition of bisectors,
the distance between $\F{x'}{\theta}$
and $\B{x}{\theta}$ is no less than $R/2$,
which implies that
$H^-(\B{x}{\theta}, \y{x}{\theta}) \subseteq H^-(\B{x'}{\theta}, \y{x'}{\theta})$.
Therefore, we can derive the following inequality
\begin{eqnarray*}
    W^*(x') &\ge&   \W{x'}{\theta}                          \\
            &=&     W(H^-(\B{x'}{\theta}, \y{x'}{\theta})   \\
            &\ge&   W(H^-(\B{x}{\theta}, \y{x}{\theta})     \\
            &=&     \W{x}{\theta}                           \\
            &=&     W^*(x),
\end{eqnarray*}
which completes the proof.
\qed
\end{proof}
}

This lemma tells us that, given a point $x$ and an angle $\theta \in MA(x)$,
all points not in $H^-(\F{x}{\theta}, \y{x}{\theta})$
can be ignored while finding \CENR{1}{1}s,
as their weight losses are no less than that of $x$.
By this lemma, we can also prove that the weight loss function
is convex along any line on the plane, as shown below.

\begin{lem} \label{lem_convexity}
    Let $x_1, x_2$ be two arbitrary distinct points on a given line $L$.
    For any point $x \in \overline{x_1x_2} \backslash \{x_1, x_2\}$,
    $W^*(x) \le \max\{W^*(x_1), W^*(x_2)\}$.
\end{lem}

\Xomit{
\begin{proof}
Suppose by contradiction that $W^*(x) > W^*(x_1)$ and $W^*(x) > W^*(x_2)$
for some point $x \in \overline{x_1x_2} \backslash \{x_1, x_2\}$.
Since $W^*(x) > W^*(x_1)$, by Lemma \ref{lem_point_MAprune}
there exists an angle $\theta \in MA(x)$ such that
$x_1$ is included in $H^-(\F{x}{\theta}, \y{x}{\theta})$.
However, since $x \in \overline{x_1x_2} \backslash \{x_1, x_2\}$,
$x_1$ and $x_2$ locate on different sides of $\F{x}{\theta}$.
It follows that $x_2$ is outside $H^-(\F{x}{\theta}, \y{x}{\theta})$
and $W^*(x_2) \ge W^*(x)$ by Lemma \ref{lem_point_MAprune},
which contradicts the assumption.
Thus, the lemma holds.
%
%
\qed
\end{proof}
}

\begin{figure}[t]
    \centering
    \begin{minipage}[b]{0.48\textwidth}
        \centering
        \includegraphics[scale=1]{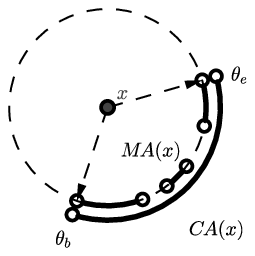}\\
        \scriptsize(a) $CA(x)$
    \end{minipage}
    \hfill
    \begin{minipage}[b]{0.48\textwidth}
        \centering
        \includegraphics[scale=1]{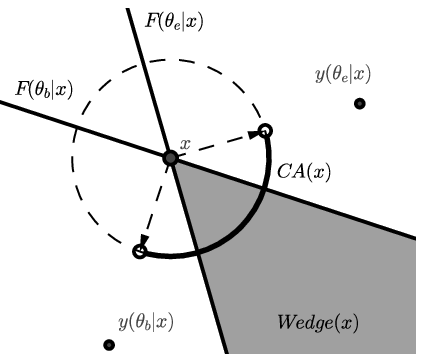}\\
        \scriptsize(b) $Wedge(x)$
    \end{minipage}
    \caption{$CA(x)$ and $Wedge(x)$.}
    \label{fig_CA_wedge}
\end{figure}

We further investigate the distribution of angles in $MA(x)$.
Let $CA(x)$ be the minimal angle interval covering all angles in $MA(x)$
(see Figure \ref{fig_CA_wedge}(a)),
and $\delta(CA(x))$ be its angle span in radians.
As mentioned before, $MA(x)$ consists of open angle interval(s) of non-zero length,
which implies that $CA(x)$ is an open interval and $\delta(CA(x)) > 0$.
Moreover, we can derive the following.

\begin{lem} \label{lem_point_ge_Pi}
    If $\delta(CA(x)) > \pi$, $x$ is a \CENR{1}{1}.
\end{lem}

\Xomit{
\begin{proof}
We prove this lemma by showing that $W^*(x') \ge W^*(x) \;\forall\; x' \ne x$.
Let $x' \in \Plane$ be an arbitrary point other than $x$,
and $\theta'$ be its polar angle with respect to $x$.
Obviously, any angle $\theta$ satisfying $x' \in H^-(\F{x}{\theta}, \y{x}{\theta})$
is in the open interval $(\theta'-\pi/2, \theta'+\pi/2)$,
the angle span of which is equal to $\pi$.
Since $\delta(CA(x)) > \pi$, by its definition there exists an angle
$\theta \in MA(x)$ such that $x' \notin H^-(\F{x}{\theta}, \y{x}{\theta})$.
Thus, by Lemma \ref{lem_point_MAprune},
we have $W^*(x') \ge W^*(x)$, thereby proves the lemma.
\qed
\end{proof}
}

We call a point $x$ satisfying Lemma \ref{lem_point_ge_Pi}
a \emph{strong \CENR{1}{1}},
since its discovery gives an
immediate solution to the \CENR{1}{1} problem.
Note that there are problem instances
in which no strong \CENR{1}{1}s exist.
%

Suppose that $\delta(CA(x)) \le \pi$ for some point $x \in \Plane$.
Let $Wedge(x)$ denote the \emph{wedge} of $x$,
defined as the intersection of the two half-planes
$H(\F{x}{\theta_b}, \y{x}{\theta_b})$
and $H(\F{x}{\theta_e}, \y{x}{\theta_e})$,
where $\theta_b$ and $\theta_e$ are the
beginning and ending angles of $CA(x)$, respectively.
As illustrated in Figure \ref{fig_CA_wedge}(b), $Wedge(x)$ is
the infinite region lying between two half-lines extending from $x$
(including $x$ and the two half-lines).
The half-lines defined by $\F{x}{\theta_e}$ and $\F{x}{\theta_b}$
are called its boundaries,
and the counterclockwise (CCW) angle between the two boundaries
is denoted by $\delta(Wedge(x))$.
Since $0 < \delta(CA(x)) \le \pi$,
we have that $Wedge(x) \ne \emptyset$ and $0 \le \delta(Wedge(x)) < \pi$.

It should be emphasized that $Wedge(x)$ is
a computational byproduct of $CA(x)$
when $x$ is not a strong \CENR{1}{1}.
In other words, not every point has its wedge.
Therefore, we make the following assumption (or restriction)
in order to avoid the misuse of $Wedge(x)$.

\begin{ass} \label{ass_CA}
    Whenever $Wedge(x)$ is mentioned,
    the point $x$ has been found not to be a strong \CENR{1}{1},
    either by computation or by properties.
    Equivalently, $\delta(CA(x)) \le \pi$.
\end{ass}

The following essential lemma makes $Wedge(x)$
our main tool for prune-and-search.
(Note that its proof cannot be trivially
derived from Lemma \ref{lem_point_MAprune},
since by definition $\theta_b$ and $\theta_e$
do not belong to the open intervals $CA(x)$ and $MA(x)$.)

\begin{lem} \label{lem_point_wedge}
    Let $x \in \Plane$ be an arbitrary point.
    For any point $x' \notin Wedge(x)$, $W^*(x') \ge W^*(x)$.
\end{lem}

\Xomit{
\begin{proof}
%
%
%
%
By symmetry, suppose that $x' \notin H(\F{x}{\theta_b}, \y{x}{\theta_b})$.
We can further divide the position of $x'$ into two cases,
(1) $x' \in H(\F{x}{\theta_e}, \y{x}{\theta_e})$ and
(2) $x' \notin H(\F{x}{\theta_e}, \y{x}{\theta_e})$.

Consider case (1).
The two assumptions ensure that there exists
an angle $\theta' \in (\theta_b, \theta_e]$,
such that $\F{x}{\theta'}$ passes through $x'$.
Obviously, any angle $\theta'' \in (\theta_b, \theta')$
satisfies that $x' \notin H(\F{x}{\theta''}, \y{x}{\theta''})$.
By the definition of $CA(x)$, there must exist an angle
$\theta'_b \in (\theta_b, \theta')$ infinitely close to $\theta_b$,
such that $\theta'_b$ belongs to $MA(x)$.
Thus, by Lemma \ref{lem_point_MAprune}, we have that $W^*(x') \ge W^*(x)$.

In case (2), for any angle $\theta'' \in MA(x)$,
we have that $x' \notin H(\F{x}{\theta''}, \y{x}{\theta''})$,
since $\theta''$ is in $(\theta_b, \theta_e)$.
Again, $W^*(x') \ge W^*(x)$ by Lemma \ref{lem_point_MAprune}.
\qed
\end{proof}
}

Finally, we consider the computation of $Wedge(x)$.

\begin{lem} \label{lem_computeWedge}
    Given a point $x \in \Plane$,
    $MA(x)$, $CA(x)$, and $Wedge(x)$ can be computed in $O(n \log n)$ time.
\end{lem}

\Xomit{
\begin{proof}
By Theorem \ref{thm_find_med}, we first compute $W^*(x)$
and those ordered tangent lines in $O(n \log n)$ time.
Then, by performing angle sweeping around $\Cr{x}$,
we can identify in $O(n)$ time those open intervals of angles $\theta$
with $\W{x}{\theta} = W^*(x)$, of which $MA(x)$ consists.
Again by sweeping around $\Cr{x}$, 
$CA(x)$ can be obtained from $MA(x)$ in $O(n)$ time.
Now, if we find $x$ to be a strong \CENR{1}{1} by checking $\delta(CA(x))$,
the \CENR{1}{1} problem is solved and the algorithm can be terminated.
Otherwise, $Wedge(x)$ can be constructed in $O(1)$ time.
\qed
\end{proof}
}

\subsection{Searching on a Line} \label{ssec_LineLocal}

Although computing wedges can be used to prune candidate points,
it does not serve as a stable prune-and-search tool,
since wedges of different points have indefinite angle intervals and spans.
However, Assumption \ref{ass_CA} makes it work fine with lines.
Here we show how to use the wedges to
compute a \emph{local optimal} point on a given line,
i.e. a point $x$ with $W^*(x) \le W^*(x')$ for any point $x'$ on the line.

Let $L$ be an arbitrary line, which is assumed
to be non-horizontal for ease of discussion.
For any point $x$ on $L$, we can compute $Wedge(x)$
and make use of it for pruning purposes
by defining its \emph{direction} with respect to $L$.
Since $\delta(Wedge(x)) < \pi$ by definition,
there are only three categories of directions
according to the intersection of $Wedge(x)$ and $L$:
\begin{description}
    \item [\rm \emph{Upward}]--
        the intersection is the half-line of $L$ above and including $x$;
    \item [\rm \emph{Downward}]--
        the intersection is the half-line of $L$ below and including $x$;
    \item [\rm \emph{Sideward}]--
        the intersection is $x$ itself.
\end{description}
If $Wedge(x)$ is sideward, $x$ is a local optimal point on $L$,
since by Lemma \ref{lem_point_wedge}
$W^*(x) \le W^*(x') \;\forall\; x' \in L$.
Otherwise, either $Wedge(x)$ is upward or downward,
the points on the opposite half of $L$
can be pruned by Lemma \ref{lem_point_wedge}.
It shows that computing wedges acts as
a predictable tool for pruning on $L$.

Next, we list sets of \emph{breakpoints} on $L$
in which a local optimal point locates.
Recall that $\mathcal{T}$ is the set of outer tangent lines
of all pairs of circles in $\mathcal{C}(V)$.
We define the \emph{$\mathcal{T}$-breakpoints} as the set $L \times \mathcal{T}$
of intersection points between $L$ and lines in $\mathcal{T}$,
and the \emph{$\mathcal{C}$-breakpoints} as the set $L \times \mathcal{C}(V)$
of intersection points between $L$ and circles in $\mathcal{C}(V)$.
We have the following lemmas for breakpoints.

\begin{lem} \label{lem_neq_weight}
    Let $x_1, x_2$ be two distinct points on $L$. If $W^*(x_1) > W^*(x_2)$,
    there exists at least a breakpoint on the segment $\overline{x_1x_2} \backslash \{x_1\}$.
\end{lem}

\Xomit{
\begin{proof}
Let $\theta$ be an arbitrary angle in $MA(x_1)$ and
$S$ be the subset of $V$ located in the half-plane
$H^-(\B{x_1}{\theta}, \y{x_1}{\theta})$.
By definition, $x_1$ is outside the convex hull $CH(\mathcal{C}(S))$ and $W^*(x_1) = W(S)$.
On the other hand, since $W^*(x_2) < W^*(x_1) = W(S)$ by assumption,
we have that $x_2$ is inside $CH(\mathcal{C}(S))$ by Lemma \ref{lem_point_roundch}.
Thus, the segment $\overline{x_1x_2} \backslash \{x_1\}$
intersects with the boundary of $CH(\mathcal{C}(S))$.
Since the boundary of $CH(\mathcal{C}(S))$ consists of
segments of lines in $\mathcal{T}$ and arcs of circles in $\mathcal{C}(V)$,
the intersection point is either a $\mathcal{T}$-breakpoint or a $\mathcal{C}$-breakpoint,
thereby proves the lemma.
\qed
\end{proof}
}

\begin{lem} \label{lem_line_candidate}
    There exists a local optimal point $x^*_L$ which is also a breakpoint.
\end{lem}

\Xomit{
\begin{proof}
Let $x^*_L$ be a local optimal point such that $W^*(x') > W^*(x^*_L)$
for some point $x'$ adjacent to $x^*_L$ on $L$.
Note that, if no such local optimal point exists,
every point on $L$ must have the same weight loss and be local optimal,
and the lemma holds trivially.
If such $x^*_L$ and $x'$ exist, by Lemma \ref{lem_neq_weight}
there is a breakpoint on $\overline{x'x^*_L} \backslash \{x'\}$,
which is $x^*_L$ itself.
Thus, the lemma holds.
\qed
\end{proof}
}

We remark that outer tangent lines parallel to $L$
are exceptional cases while considering breakpoints.
For any line $T \in \mathcal{T}$ that is parallel to $L$,
either $T$ does not intersect with $L$ or they just coincide.
In either case, $T$ is irrelevant to the finding of local optimal points,
and should not be counted for defining $\mathcal{T}$-breakpoints.

Now, by Lemma \ref{lem_line_candidate}, if we have all breakpoints on $L$
sorted in the decreasing order of their y-coordinates,
a local optimal point can be found by performing binary search using wedges.
Obviously, such sorted sequence can be obtained in $O(n^2 \log n)$ time,
since $|L \times \mathcal{T}| = O(n^2)$ and $|L \times \mathcal{C}(V)| = O(n)$.
However, in order to speed up the computations
of local optimal points on multiple lines,
alternatively we propose an $O(n^2 \log n)$-time preprocessing,
so that a local optimal point on any given line
can be computed in $O(n \log^2 n)$ time.

The preprocessing itself is very simple.
For each point $v \in V$, we compute a sequence $P(v)$, consisting of points in
$V \backslash \{v\}$ sorted in increasing order
of their polar angles with respect to $v$.
The computation for all $v \in V$ takes $O(n^2 \log n)$ time in total.
Besides, all outer tangent lines in $\mathcal{T}$ are computed in $O(n^2)$ time.
We will show that, for any given line $L$,
$O(n)$ sorted sequences can be obtained from
these pre-computed sequences in $O(n \log n)$ time,
which can be used to replace a sorted sequence of
all $\mathcal{T}$-breakpoints in the process of binary search.

For any two points $v \in V$ and $z \in \Plane$,
let $\Tr{v}{z}$ be the outer tangent line
of $\Cr{v}$ and $\Cr{z}$ to the right of the line from $v$ to $z$.
Similarly, let $\Tl{v}{z}$ be the outer tangent line to the left.
(See Figure \ref{fig_outertangent}.)
Moreover, let $\tr[L]{v}{z}$ and $\tl[L]{v}{z}$ be the points
at which $\Tr{v}{z}$ and $\Tl{v}{z}$ intersect with $L$, respectively.
We partition $\mathcal{T}$ into $O(n)$ sets
$\mathcal{T}^r(v) = \{\Tr{v}{v_i}|v_i \in V \backslash \{v\}\}$ 
and $\mathcal{T}^l(v) = \{\Tl{v}{v_i}|v_i \in V \backslash \{v\}\}$ for $v \in V$,
and consider their corresponding $\mathcal{T}$-breakpoints independently.
By symmetry, we only discuss the case about $L \times \mathcal{T}^r(v)$.

\begin{figure}[t]
    \centering
    \includegraphics[scale=1]{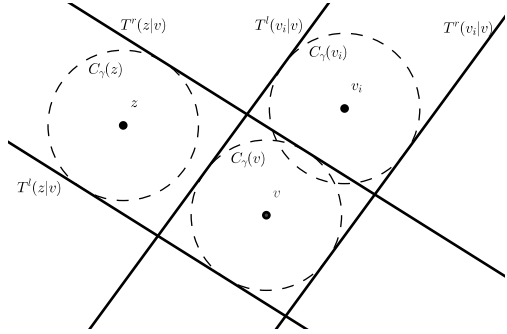}
    \caption{Outer tangent lines of $v$.}
    \label{fig_outertangent}
\end{figure}

\begin{lem} \label{lem_buildsequence}
    For each $v \in V$, we can compute $O(1)$ sequences of $\mathcal{T}$-breakpoints on $L$,
    which satisfy the following conditions:
    \begin{enumerate} [(a)]
        \item Each sequence is of length $O(n)$ and can be obtained in $O(\log n)$ time.
        \item Breakpoints in each sequence are sorted in decreasing y-coordinates.
        \item The union of breakpoints in all sequences form $L \times \mathcal{T}^r(v)$.
    \end{enumerate}
\end{lem}

\Xomit{
\begin{proof}
Without loss of generality, suppose that $v$ is
either strictly to the right of $L$ or on $L$.
Note that each point $v_i \in V \backslash \{v\}$
corresponds to exactly one outer tangent line $\Tr{v}{v_i}$,
thereby exactly one breakpoint $\tr[L]{v}{v_i}$.
Such one-to-one correspondence can be easily done in $O(1)$ time.
Therefore, equivalently we are computing sequences of
points in $V \backslash \{v\}$, instead of breakpoints.

In the following, we consider two cases about
the relative position between $L$ and $\Cr{v}$,
(1) $L$ intersects with $\Cr{v}$ at zero or one point,
(2) $L$ intersects with $\Cr{v}$ at two points.

\begin{description}
    \item [Case (1):]
        Let $\theta_L$ be the angle of the upward direction along $L$.
        See Figure \ref{fig_LCintersect}(a).
        We classify the points in $V \backslash \{v\}$
        by their polar angles with respect to $v$.
        Let $P_1(v)$ denote the sequence of those points with polar angles
        in the interval $(\theta_L, \theta_L + \pi)$ and sorted in CCW order.
        Similarly, let $P_2(v)$ be the sequence of points with polar angles
        in $(\theta_L + \pi, \theta_L)$ and sorted in CCW order.
        Obviously, $P_1(v)$ and $P_2(v)$ together satisfy condition (c).
        (Note that points with polar angles $\theta_L$ and $\theta_L + \pi$ are ignored,
        since they correspond to outer tangent lines parallel to $L$.)

        By general position assumption, we can observe that,
        for any two distinct points $v_i, v_j$ in $P_1(v)$,
        $\tr[L]{v}{v_i}$ is strictly above $\tr[L]{v}{v_j}$
        if and only if $v_i$ precedes $v_j$ in $P_1(v)$.
        Thus, the ordering of points in $P_1(v)$ implicitly describes an ordering
        of their corresponding breakpoints in decreasing y-coordinates.
        Similarly, the ordering in $P_2(v)$ implies an ordering
        of corresponding breakpoints in decreasing y-coordinates.
        It follows that both $P_1(v)$ and $P_2(v)$ satisfy condition (b).

        As for condition (a), both $P_1(v)$ and $P_2(v)$ are of length $O(n)$ by definition.
        Also, since we have pre-computed the sequence $P(v)$
        as all points in $V \backslash \{v\}$ sorted in CCW order,
        $P_1(v)$ and $P_2(v)$ can be implicitly represented
        as concatenations of subsequences of $P(v)$.
        This can be done in $O(\log n)$ time by searching in $P(v)$ the foremost elements
        with polar angles larger than $\theta_L$ and $\theta_L + \pi$, respectively.

\begin{figure}[t]
    \centering
    \begin{minipage}[b]{0.48\textwidth}
        \centering
        \includegraphics[scale=0.75]{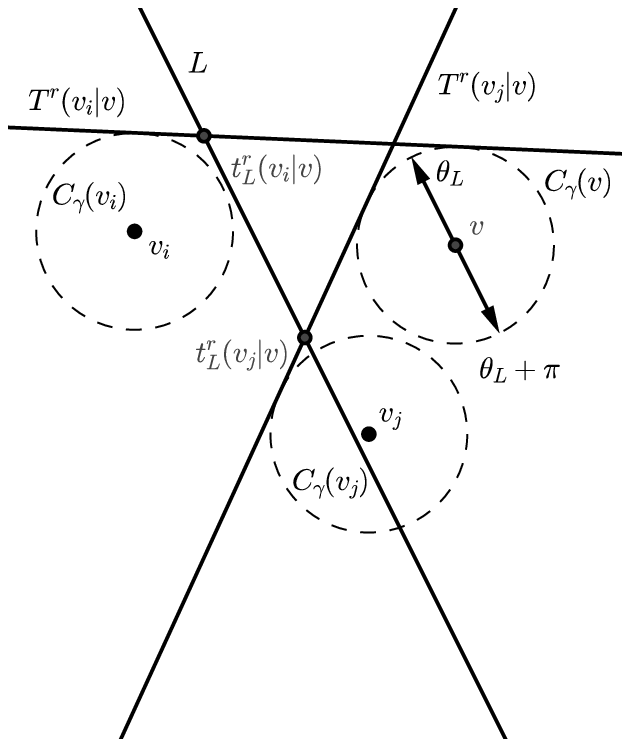}\\
        \scriptsize(a) no intersection
    \end{minipage}
    \hfill
    \begin{minipage}[b]{0.48\textwidth}
        \centering
        \includegraphics[scale=0.75]{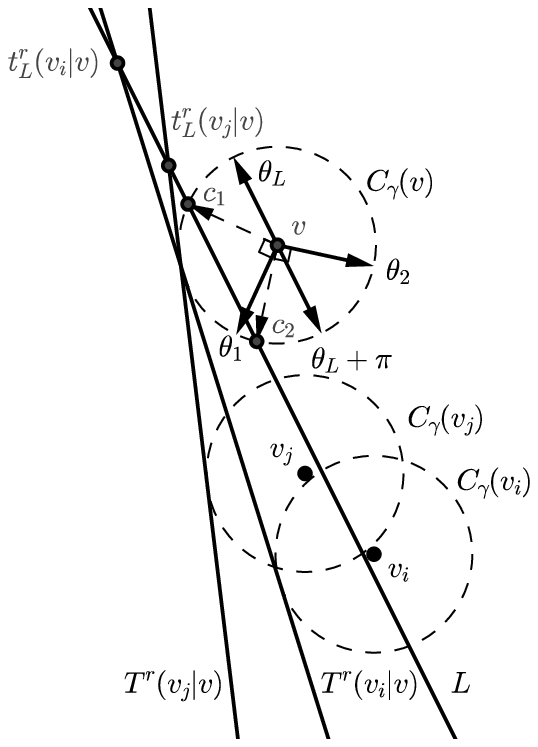}\\
        \scriptsize(b) two intersection points
    \end{minipage}
    \caption{Two subcases about how $\Cr{v}$ intersects $C$.}
    \label{fig_LCintersect}
\end{figure}

    \item [Case (2):]
        Suppose that the two intersection points between $L$ and $\Cr{v}$
        are $c_1$ and $c_2$, where $c_1$ is above $c_2$.
        Let $\theta_1 = \theta'_1 + \pi/2$ and $\theta_2 = \theta'_2 + \pi/2$,
        in which $\theta'_1$ and $\theta'_2$ are respectively
        the polar angles of $c_1$ and $c_2$ with respect to $v$.
        See Figure \ref{fig_LCintersect}(b).
        By assumption, we have that
        $\theta_L < \theta'_1 < \theta_L + \pi/2 < \theta'_2 < \theta_L + \pi$,
        which implies that $\theta_1 \in (\theta_L, \theta_L + \pi)$
        and $\theta_2 \in (\theta_L + \pi, \theta_L)$.

        We divide the points in $V \backslash \{v\}$ into four sequences
        $P_1(v)$, $P_2(v)$, $P_3(v)$, and $P_4(v)$ by their polar angles with respect to $v$.
        $P_1(v)$ consists of points with polar angles in $(\theta_L, \theta_1)$,
        $P_2(v)$ in $[\theta_1, \theta_L + \pi)$, $P_3(v)$ in $(\theta_L + \pi, \theta_2]$,
        and $P_4(v)$ in $(\theta_2, \theta_L)$, all sorted in CCW order.
        It follows that the four sequences satisfy conditions (c).

        Condition (a) and (b) hold for $P_1(v)$ and $P_4(v)$
        from similar discussion as above.
        However, for any two distinct points $v_i, v_j$ in $P_2(v)$,
        we can observe that $\tr[L]{v}{v_i}$ is strictly below $\tr[L]{v}{v_j}$
        if and only if $v_i$ precedes $v_j$ in $P_2(v)$.
        Similarly, the argument holds for $P_3(v)$.
        Thus, what satisfy condition (b) are actually
        the reverse sequences of $P_2(v)$ and $P_3(v)$,
        which can also be obtained in $O(\log n)$ time, satisfying condition (a).
        \qed
\end{description}
\end{proof}
}

By Lemma \ref{lem_buildsequence}.(c), searching in $L \times \mathcal{T}^r(v)$
is equivalent to searching in the $O(1)$ sequences of breakpoints,
which can be computed more efficiently than the obvious way.
Besides, we can also obtain a symmetrical lemma
constructing sequences for $L \times \mathcal{T}^l(v)$.
In the following, we show how to perform a binary search within these sequences.

\begin{lem} \label{lem_localopt}
    With an $O(n^2 \log n)$-time preprocessing, given an arbitrary line $L$,
    a local optimal point $x^*_L$ can be computed in $O(n \log^2 n)$ time.
\end{lem}

\Xomit{
\begin{proof}
By Lemma \ref{lem_line_candidate}, the searching of $x^*_L$
can be done within $L \times \mathcal{T}$ and $L \times \mathcal{C}(V)$.
$L \times \mathcal{T}$ can be further divided into
$L \times \mathcal{T}^r(v)$ and $L \times \mathcal{T}^l(v)$ for each $v \in V$.
By Lemma \ref{lem_buildsequence}, these $2n$ sets can be replaced
by $O(n)$ sorted sequences of breakpoints on $L$.
Besides, $L \times \mathcal{C}(V)$ consists of no more than $2n$ breakpoints,
which can be computed and arranged into a
sorted sequence in decreasing y-coordinates.
Therefore, we can construct $N_0 = O(n)$ sequences $P_1, P_2, \cdots, P_{N_0}$ of breakpoints,
each of length $O(n)$ and sorted in decreasing y-coordinates.

The searching in the $N_0$ sorted sequences is done by
performing parametric search for parallel binary searches,
introduced in \cite{Cole-87}.
The technique we used here is similar to the algorithm in \cite{Cole-87},
but uses a different weighting scheme.
For each sorted sequence $P_j$, $1 \le j \le N_0$,
we first obtain its middle element $x_j$,
and associate $x_j$ with a weight $m_j$
equal to the number of elements in $P_j$.
Then, we compute the \emph{weighted median} \cite{Reiser-78}
of the $N_0$ middle elements, defined as the element $x$ such that
$\sum \{m_j|x_j\text{ is above }x\} \ge \sum m_j/2$
and $\sum \{m_j|x_j\text{ is below }x\} \ge \sum m_j/2$.
Finally, we apply Lemma \ref{lem_computeWedge} on the point $x$.
If $x$ is a strong \CENR{1}{1}, of course it is local optimal.
If not, Assumption \ref{ass_CA} holds and $Wedge(x)$ can be computed.
If $Wedge(x)$ is sideward,
a local optimal point $x^*_L = x$ is directly found.
Otherwise, $Wedge(x)$ is either upward or downward,
and thus all breakpoints on the opposite half
can be pruned by Lemma \ref{lem_point_wedge}.
The pruning makes a portion of sequences,
that possesses over half of total breakpoints
by the definition of weighted median,
lose at least a quarter of their elements.
Hence, at least one-eighths of breakpoints are pruned.
By repeating the above process,
we can find $x^*_L$ in at most $O(\log n)$ iterations.

The time complexity for finding $x^*_L$ is analyzed as follows.
By Lemma \ref{lem_buildsequence}, constructing sorted sequences for
$L \times \mathcal{T}^r(v)$ and $L \times \mathcal{T}^l(v)$
for all $v \in V$ takes $O(n \log n)$ time.
Computing and sorting $L \times \mathcal{C}(V)$ also takes $O(n \log n)$ time.
There are at most $O(\log n)$ iterations of the pruning process.
At each iteration, the $N_0$ middle elements and
their weighted median $x$ can be obtained in $O(N_0) = O(n)$ time
by the linear-time weighted selection algorithm \cite{Reiser-78}.
Then, the computation of $Wedge(x)$ takes
$O(n \log n)$ time by Lemma \ref{lem_computeWedge}.
Finally, the pruning of those sequences can be done in $O(n)$ time.
In summary, the searching of $x^*_L$ requires
$O(n \log n) + O(\log n) \times O(n \log n) = O(n \log^2 n)$ time.
\qed
\end{proof}
}

We remark that, by Lemma \ref{lem_localopt}, it is easy to obtain
an intermediate result for the \CENR{1}{1} problem on the plane.
By Lemma \ref{lem_TTCCTC}, there exists a \CENR{1}{1} in $\mathcal{T} \times \mathcal{T}$,
$\mathcal{T} \times \mathcal{C}(V)$, and $\mathcal{C}(V) \times \mathcal{C}(V)$.
By applying Lemma \ref{lem_localopt} to the $O(n^2)$ lines in $\mathcal{T}$,
the local optimum among the intersection points in
$\mathcal{T} \times \mathcal{T}$ and $\mathcal{T} \times \mathcal{C}(V)$
can be obtained in $O(n^3 \log^2 n)$ time.
By applying Theorem \ref{thm_find_med} on the $O(n^2)$
intersection points in $\mathcal{C}(V) \times \mathcal{C}(V)$,
the local optimum among them can be obtained in $O(n^3 \log n)$ time.
Thus, we can find a \CENR{1}{1} in $O(n^3 \log^2 n)$ time,
a nearly $O(n^2)$ improvement over the $O(n^5 \log n)$ bound in \cite{Drezner-82}.



\section{$(1|1)_R$-Centroid on the Plane} \label{sec_Alg_plane}

In this section, we study the \CENR{1}{1} problem
and propose an improved algorithm of time complexity $O(n^2 \log n)$.
This algorithm is as efficient as the best-so-far algorithm for the \CEN{1}{1} problem,
but based on a completely different approach.

In Subsection \ref{ssec_Line}, we extend the algorithm of Lemma \ref{lem_localopt}
to develop a procedure allowing us to prune
candidate points with respect to a given vertical line.
Then, in Subsection \ref{ssec_PlaneSearch}, we show how to compute a \CENR{1}{1}
in $O(n^2 \log n)$ time based on this newly-developed pruning procedure.

\subsection{Pruning with Respect to a Vertical Line} \label{ssec_Line}

Let $L$ be an arbitrary vertical line on the plane.
We call the half-plane strictly to the left of $L$
the \emph{left plane} of $L$
and the one strictly to its right the \emph{right plane} of $L$.
A sideward wedge of some point on $L$ is
said to be \emph{rightward} (\emph{leftward})
if it intersects the right (left) plane of $L$.
We can observe that, if there is some point $x \in L$
such that $Wedge(x)$ is rightward,
every point $x'$ on the left plane of $L$ can be pruned,
since $W^*(x') \ge W^*(x)$ by Lemma \ref{lem_point_wedge}.
Similarly, if $Wedge(x)$ is leftward,
points on the right plane of $L$ can be pruned.
Although the power of wedges is not fully exerted in this way,
pruning via vertical lines and sideward wedges is superior than
directly via wedges due to predictable pruning regions.

Therefore, in this subsection we describe
how to design a procedure that enables us to prune
either the left or the right plane of a given vertical line $L$.
As mentioned above, the key point is the searching of sideward wedges on $L$.
It is achieved by carrying out three conditional phases.
In the first phase, we try to find some proper breakpoints with sideward wedges.
If failed, we pick some representative point in the second phase
and check its wedge to determine whether or not sideward wedges exist.
Finally, in case of their nonexistence,
we show that their functional alternative can be computed,
called the \emph{pseudo wedge},
that still allows us to prune the left or right plane of $L$.
In the following, we develop a series of lemmas
to demonstrate the details of the three phases.


\begin{pro} \label{pro_wedge_dir}
   Given a point $x \in L$, for each possible direction of $Wedge(x)$,
   the corresponding $CA(x)$ satisfies the following conditions:
    \begin{description}
        \item [\rm Upward]-- $CA(x) \subseteq (0, \pi)$,
        \item [\rm Downward]-- $CA(x) \subseteq (\pi, 2\pi)$,
        \item [\rm Rightward]-- $0 \in CA(x)$,
        \item [\rm Leftward]-- $\pi \in CA(x)$.
    \end{description}
\end{pro}

\Xomit{
\begin{proof}
When $Wedge(x)$ is upward,
by definition the beginning angle $\theta_b$
and the ending angle $\theta_e$ of $CA(x)$
must satisfy that both half-planes
$H(\F{x}{\theta_b}, \y{x}{\theta_b})$
and $H(\F{x}{\theta_e}, \y{x}{\theta_e})$
include the half-line of $L$ above $x$.
It follows that $0 \le \theta_b, \theta_e \le \pi$,
and thus $CA(x) \subseteq (0, \pi)$.
(Recall that $\theta_b, \theta_e \notin CA(x)$.)
The case that $Wedge(x)$ is downward
can be proved in a symmetric way.

When $Wedge(x)$ is rightward, we can see that
$H(\F{x}{\theta_b}, \y{x}{\theta_b})$ must not
contain the half-line of $L$ above $x$,
and thus $\pi < \theta_b < 2\pi$.
By similar arguments, $0 < \theta_e < \pi$.
Therefore, counterclockwise covering angles from $\theta_b$ to $\theta_e$,
$CA(x)$ must include the angle $0$.
The case that $Wedge(x)$ is leftward
can be symmetrically proved.
\qed
\end{proof}
}

\begin{lem} \label{lem_weight_shift}
    Let $x_1, x_2$ be two points on $L$, where $x_1$ is strictly above $x_2$.
    For any angle $0 \le \theta \le \pi$, $\W{x_1}{\theta} \le \W{x_2}{\theta}$.
    Symmetrically, for $\pi \le \theta \le 2\pi$, $\W{x_2}{\theta} \le \W{x_1}{\theta}$.
\end{lem}

\Xomit{
\begin{proof}
For any angle $0 \le \theta \le \pi$, we can observe that
$H^-(\B{x_1}{\theta}, \y{x_1}{\theta}) \subset H^-(\B{x_2}{\theta}, \y{x_2}{\theta})$,
since $x_1$ is strictly above $x_2$.
It follows that $\W{x_1}{\theta} \le \W{x_2}{\theta}$.
The second claim also holds by symmetric arguments.
\qed
\end{proof}
}

\begin{lem} \label{lem_weight_trans}
    Let $x$ be an arbitrary point on $L$.
    If $Wedge(x)$ is either upward or downward,
    for any point $x' \in L \backslash Wedge(x)$,
    $Wedge(x')$ has the same direction as $Wedge(x)$.
\end{lem}

\Xomit{
\begin{proof}
By symmetry, we prove that, if $Wedge(x)$ is upward,
$Wedge(x')$ is also upward for every $x' \in L$ strictly below $x$.
By Property \ref{pro_wedge_dir}, the fact that $Wedge(x)$ is upward
means that $CA(x) \subset (0, \pi)$ and thus $MA(x) \subset (0, \pi)$.
Let $x'$ be a point on $L$ strictly below $x$.
By Lemma \ref{lem_weight_shift},
we have that $\W{x'}{\theta} \ge \W{x}{\theta}$ for $0 < \theta < \pi$
and $\W{x'}{\theta} \le \W{x}{\theta}$ for $\pi \le \theta \le 2\pi$.
It follows that $MA(x') \subset (0, \pi)$ and $CA(x') \subset (0, \pi)$,
so $Wedge(x')$ is upward as well.
\qed
\end{proof}
}

Following from this lemma, if there exist two arbitrary points
$x_1$ and $x_2$ on $L$ with their wedges downward and upward, respectively,
we can derive that $x_1$ must be strictly above $x_2$,
and that points with sideward wedges or even strong \CENR{1}{1}s
can locate only between $x_1$ and $x_2$.
Thus, we can find sideward wedges between
some specified downward and upward wedges.
Let $x_D$ be the lowermost breakpoint on $L$ with its wedge downward,
$x_U$ the uppermost breakpoint on $L$ with its wedge upward,
and $G_{DU}$ the open segment $\overline{x_Dx_U} \backslash \{x_D, x_U\}$.
(For ease of discussion, we assume that both $x_D$ and $x_U$ exist on $L$,
and show how to resolve this assumption later by constructing a bounded box.)
Again, $x_D$ is strictly above $x_U$.
Also, we have the following corollary by their definitions.

\begin{cor} \label{cor_breakpoint}
    If there exist breakpoints in the segment $G_{DU}$,
    for any such breakpoint $x$, either $x$ is a strong \CENR{1}{1}
    or $Wedge(x)$ is sideward.
\end{cor}

Given $x_D$ and $x_U$, the first phase can thus be done by
checking whether there exist breakpoints in $G_{DU}$
and picking any of them if exist.
Supposing that the picked one is not a strong \CENR{1}{1},
a sideward wedge is found by Corollary \ref{cor_breakpoint}
and can be used for pruning.
Notice that, when there are two or more such breakpoints,
one may question whether their wedges are of the same direction,
as different directions result in inconsistent pruning results.
The following lemma answers the question in the positive.



\begin{lem} \label{lem_same_side}
    Let $x_1$, $x_2$ be two distinct points on $L$,
    where $x_1$ is strictly above $x_2$
    and none of them is a strong \CENR{1}{1}.
    If $Wedge(x_1)$ and $Wedge(x_2)$ are both sideward,
    they are either both rightward or both leftward.
\end{lem}

\Xomit{
\begin{proof}
We prove this lemma by contradiction.
By symmetry, suppose the case that
$Wedge(x_1)$ is rightward and $Wedge(x_2)$ is leftward.
This case can be further divided into two subcases
by whether or not $CA(x_1)$ and $CA(x_2)$ intersect.

Consider first that $CA(x_1)$ does not intersect $CA(x_2)$.
Because $Wedge(x_1)$ is rightward,
$0 \in CA(x_1)$ by Property \ref{pro_wedge_dir}.
Thus, there exists an angle $\theta$, $0 < \theta < \pi$,
such that $\theta \in MA(x_1)$.
Since $x_1$ is strictly above $x_2$, by Lemma \ref{lem_weight_shift}
we have that $W^*(x_1) = \W{x_1}{\theta} \le \W{x_2}{\theta} \le W^*(x_2)$.
Furthermore, since $Wedge(x_2)$ is leftward,
we can see that $x_1 \notin Wedge(x_2)$ and therefore
$W^*(x_1) \ge W^*(x_2)$ by Lemma \ref{lem_point_wedge}.
It follows that $\W{x_2}{\theta} = W^*(x_2)$ and thus $\theta \in MA(x_2)$.
By definition, $MA(x_1) \subseteq CA(x_1)$ and $MA(x_2) \subseteq CA(x_2)$,
which implies that $CA(x_1)$ and $CA(x_2)$ intersect at $\theta$,
contradicting the subcase assumption.

When $CA(x_1)$ intersects $CA(x_2)$,
their intersection must be completely included in either $(0, \pi)$ or $(\pi, 2\pi)$
due to Assumption \ref{ass_CA}.
By symmetry, we assume the latter subcase.
Using similar arguments as above, we can find an angle $\theta'$,
where $0 < \theta' < \pi$, such that $\theta' \in MA(x_1)$ and $\theta' \in MA(x_2)$.
This is a contradiction, since $\theta' \notin (\pi, 2\pi)$.

Since both subcases do not hold, the lemma is proved.
\qed
\end{proof}
}

The second phase deals with the case that
no breakpoint exists between $x_D$ and $x_U$
by determining the wedge direction of an arbitrary inner point in $G_{DU}$.
We begin with several auxiliary lemmas.



\begin{lem} \label{lem_MA_change}
    Let $x_1, x_2$ be two distinct points on $L$
    such that $W^*(x_1) = W^*(x_2)$ and $x_1$ is strictly above $x_2$.
    There exists at least one breakpoint in the segment
    \begin{enumerate} [(a)]
        \item $\overline{x_1x_2} \backslash \{x_2\}$,
            if $MA(x_2)$ intersects $(0, \pi)$ but $MA(x_1)$ does not,
        \item $\overline{x_1x_2} \backslash \{x_1\}$,
            if $MA(x_1)$ intersects $(\pi, 2\pi)$ but $MA(x_2)$ does not.
    \end{enumerate}
\end{lem}

\Xomit{
\begin{proof}
By symmetry, we only show the correctness of condition (a).
From its assumption, there exists an angle $\theta$, where
$0 < \theta < \pi$, such that $\theta \in MA(x_2)$.
Let $S = V \bigcap H^-(\B{x_2}{\theta}, \y{x_2}{\theta})$.
By definition, we have that $W(S) = \W{x_2}{\theta} = W^*(x_2) = W^*(x_1)$
and $CH(\mathcal{C}(S)) \subset H^-(\F{x_2}{\theta}, \y{x_2}{\theta})$,
which implies that $CH(\mathcal{C}(S))$ is strictly above $\F{x_2}{\theta}$.
(See Figure \ref{fig_CH_intersect_L}.)

\begin{figure}[t]
    \centering
    \includegraphics[scale=0.75]{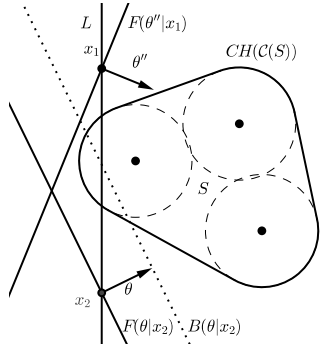}
    \caption{$CH(\mathcal{C}(S))$ intersects $L$ between $x_1$ and $x_2$.}
    \label{fig_CH_intersect_L}
\end{figure}

We first claim that $CH(\mathcal{C}(S))$ intersects $L$.
If not, there must exist an angle $\theta'$,
where $0 < \theta' < \pi$, such that
$CH(\mathcal{C}(S)) \subset H^-(\F{x_1}{\theta'}, \y{x_1}{\theta'})$,
that is, $S \subset H^-(\B{x_1}{\theta'}, \y{x_1}{\theta'})$.
By definition, $W^*(x_1) \ge \W{x_1}{\theta'} \ge W(S)$.
Since $W^*(x_1) = W(S)$, $\W{x_1}{\theta'} = W^*(x_1)$
and thus $\theta' \in MA(x_1)$,
which contradicts the condition that
$MA(x_1)$ does not intersect $(0, \pi)$.
Thus, the claim holds.

When $CH(\mathcal{C}(S))$ intersects $L$,
$x_1$ locates either inside or outside $CH(\mathcal{C}(S))$.
Since $x_2$ locates outside $CH(\mathcal{C}(S))$,
in the former case the boundary of $CH(\mathcal{C}(S))$
intersects $\overline{x_1x_2} \backslash \{x_2\}$
and forms a breakpoint, thereby proves condition (a).
On the other hand, if $x_1$ is outside $CH(\mathcal{C}(S))$,
again there exists an angle $\theta''$ such that
$CH(\mathcal{C}(S)) \subset H^-(\F{x_1}{\theta''}, \y{x_1}{\theta''})$.
By similar arguments, we can show that $\theta'' \in MA(x_1)$.
By assumption, $\theta''$ must belong to $(\pi, 2\pi)$,
which implies that $CH(\mathcal{C}(S))$ is strictly below $\F{x_1}{\theta''}$.
Since $CH(\mathcal{C}(S))$ is strictly above $\F{x_2}{\theta}$ as mentioned,
any intersection point between $CH(\mathcal{C}(S))$ and $L$
should be inner to $\overline{x_1x_2}$.
Therefore, the lemma holds.
\qed
\end{proof}
}

\begin{lem} \label{lem_same_loss}
    Let $G$ be a line segment connecting two consecutive breakpoints on $L$.
    For any two distinct points $x_1, x_2$ inner to $G$, $W^*(x_1) = W^*(x_2)$.
\end{lem}

\Xomit{
\begin{proof}
Suppose to the contrary that $W^*(x_1) \neq W^*(x_2)$.
By Lemma \ref{lem_neq_weight},
there exists at least one breakpoint in $\overline{x_1x_2}$,
which contradicts the definition of $G$.
Thus, the lemma holds.
\qed
\end{proof}
}

\begin{lem} \label{lem_no_breakpoint}
    When there is no breakpoint between $x_D$ and $x_U$,
    any two distinct points $x_1, x_2$ in $G_{DU}$
    have the same wedge direction,
    if they are not strong \CENR{1}{1}s.
\end{lem}

\Xomit{
\begin{proof}
Suppose by contradiction that the directions
of their wedges are different.
By Lemmas \ref{lem_weight_trans} and \ref{lem_same_side},
there are only two possible cases. 
\begin{enumerate} [(1)]
    \item $Wedge(x_1)$ is downward,
    and $Wedge(x_2)$ is either sideward or upward.
    \item $Wedge(x_1)$ is sideward,
    and $Wedge(x_2)$ is upward.
\end{enumerate}
In the following, we show that both cases do not hold.

\begin{description}
    \item [Case (1):]
        Because $Wedge(x_1)$ is downward, we have that
        $CA(x_1) \subseteq (\pi, 2\pi)$ by Property \ref{pro_wedge_dir}
        and thus $MA(x_1)$ does not intersect $(0, \pi)$.
        On the other hand, whether $Wedge(x_2)$ is sideward or upward,
        we can see that $CA(x_2)$ and $MA(x_2)$ intersect $(0, \pi)$
        by again Property \ref{pro_wedge_dir}.
        Since $W^*(x_1) = W^*(x_2)$ by Lemma \ref{lem_same_loss}, 
        the status of the two points satisfies
        the condition (a) of Lemma \ref{lem_MA_change},
        so at least one breakpoint exists between $x_1$ and $x_2$.
        By definitions of $x_1$ and $x_2$, this breakpoint is inner to $G_{DU}$,
        thereby contradicts the assumption.
        Therefore, Case (1) does not hold.
    \item [Case (2):]
        The proof of Case (2) is symmetric to that of Case (1).
        The condition (b) of Lemma \ref{lem_MA_change} can be applied similarly
        to show the existence of at least one breakpoint between $x_1$ and $x_2$,
        again a contradiction.
\end{description}
Combining the above discussions,
we prove that the wedges of $x_1$ and $x_2$ are of the same direction,
thereby completes the proof of this lemma.
\qed
\end{proof}
}

This lemma enables us to pick an arbitrary point in $G_{DU}$,
e.g., the bisector point $x_B$ of $x_D$ and $x_U$,
as the representative of all inner points in $G_{DU}$.
If $x_B$ is not a strong \CENR{1}{1} and $Wedge(x_B)$ is sideward,
the second phase finishes with a sideward wedge found.
Otherwise, if $Wedge(x_B)$ is downward or upward,
we can derive the following and have to invoke the third phase.

\begin{lem} \label{lem_no_strong}
    If there is no breakpoint between $x_D$ and $x_U$
    and $Wedge(x_B)$ is not sideward,
    there exist neither strong \CENR{1}{1}s
    nor points with sideward wedges on $L$.
\end{lem}

\Xomit{
\begin{proof}
By Lemma \ref{lem_weight_trans}, this lemma holds for points not in $G_{DU}$.
Without loss of generality, suppose that $Wedge(x_B)$ is downward.
For all points in $G_{DU}$ above $x_B$,
the lemma holds by again Lemma \ref{lem_weight_trans}.

Consider an arbitrary point $x \in \overline{x_Bx_U} \backslash \{x_B,x_U\}$.
We first show that $x$ is not a strong \CENR{1}{1}.
Suppose to the contrary that $x$ really is.
By definition, we have that $\delta(CA(x)) > \pi$,
and thus $CA(x)$ and $MA(x)$ intersect $(0, \pi)$.
On the other hand, $CA(x_B)$ and $MA(x_B)$ do not intersect $(0, \pi)$
due to downward $Wedge(x_B)$ and Property \ref{pro_wedge_dir}.
Since $W^*(x_B) = W^*(x)$ by Lemma \ref{lem_same_loss},
applying the condition (a) of Lemma \ref{lem_MA_change} to $x_B$ and $x$
shows that at least one breakpoint exists between them,
which contradicts the no-breakpoint assumption.
Now that $x$ is not a strong \CENR{1}{1}, it must have a downward wedge,
as $x_B$ does by Lemma \ref{lem_no_breakpoint}.
Therefore, the lemma holds for all points on $L$.
\qed
\end{proof}
}

When $L$ satisfies Lemma \ref{lem_no_strong},
it consists of only points with downward or upward wedges,
and is said to be \emph{non-leaning}.
Obviously, our pruning strategy via sideward wedges
could not apply to such non-leaning lines.
The third phase overcomes this obstacle
by constructing a functional alternative of sideward wedges,
called the pseudo wedge, on either $x_D$ or $x_U$,
so that pruning with respect to $L$ is still achievable.
Again, we start with auxiliary lemmas.

%

\begin{lem} \label{lem_no_sideward}
    If $L$ is non-leaning, the following statements hold:
    \begin{enumerate} [(a)]
        \item $W^*(x_D) \neq W^*(x_U)$,
        \item $W^*(x) = \max\{W^*(x_D), W^*(x_U)\}$
            for all points $x \in G_{DU}$.
    \end{enumerate}
\end{lem}

\Xomit{
\begin{proof}
We prove the correctness of statement (a) by contradiction,
and suppose that $W^*(x_D) = W^*(x_U)$.
Besides, the fact that $L$ is non-leaning
implies that no breakpoint exists in $G_{DU}$.
By Lemmas \ref{lem_no_strong} and \ref{lem_no_breakpoint},
the wedges of all points in $G_{DU}$
are of the same direction, either downward or upward.
Suppose the downward case by symmetry,
and pick an arbitrary point in $G_{DU}$, say, $x_B$.
Since $Wedge(x_B)$ is downward,
we have that $MA(x_B)$ does not intersect $(0, \pi)$.
Oppositely, by definition $Wedge(x_U)$ is upward,
so $CA(x_U)$ and $MA(x_U)$ are included in $(0, \pi)$.
Because $x_B$ is strictly above $x_U$ and $W^*(x_D) = W^*(x_U)$,
according to the condition (a) of Lemma \ref{lem_MA_change},
there exists at least one breakpoint in $\overline{x_Bx_U} \backslash \{x_U\}$,
which is a contradiction.
Therefore, statement (a) holds.

The proof of statement (b) is also done by contradiction.
By symmetry, assume that $W^*(x_D) > W^*(x_U)$ in statement (a).
Consider an arbitrary point $x \in G_{DU}$.
By Lemma \ref{lem_convexity}, we have that
$W^*(x) \le \max\{W^*(x_D), W^*(x_U)\} = W^*(x_D)$.
Suppose that the equality does not hold.
Then, by Lemma \ref{lem_neq_weight}, at least one breakpoint
exists in the segment $\overline{x_Dx} \backslash \{x_D\}$,
contradicting the no-breakpoint fact.
Thus, $W^*(x) = W^*(x_D)$ and statement (b) holds.
\qed
\end{proof}
}

Let $W_1 = \max\{W^*(x_D), W^*(x_U)\}$.
We are going to define the pseudo wedge on either $x_U$ or $x_D$,
depending on which one has the smaller weight loss.
We consider first the case that $W^*(x_D) > W^*(x_U)$,
and obtain the following.

%
%

\begin{lem} \label{lem_xu_boundary}
    If $L$ is non-leaning and $W^*(x_D) > W^*(x_U)$,
    there exists one angle $\theta$ for $x_U$,
    where $\pi \le \theta \le 2\pi$, such that 
    $W(H(\B{x_U}{\theta}, \y{x_U}{\theta})) \ge W_1$.
\end{lem}

\Xomit{
\begin{proof}
We first show that there exists at least
a subset $S \subseteq V$ with $W(S) = W_1$,
such that $x_U$ locates on the upper boundary of $CH(\mathcal{C}(S))$.
Let $x$ be the point strictly above but arbitrarily close to $x_U$ on $L$.
By Lemma \ref{lem_no_sideward}, $W^*(x) = W_1$,
hence $W^*(x) > W^*(x_U)$ by case assumption.
It follows that $x_U \in Wedge(x)$ by Lemma \ref{lem_point_wedge}
and $Wedge(x)$ must be downward.
By Property \ref{pro_wedge_dir},
we have that $CA(x) \subseteq (\pi, 2\pi)$.
Thus, there exists an angle $\theta' \in MA(x)$,
where $\pi < \theta' < 2\pi$,
such that $W(H^-(\B{x}{\theta'}, \y{x}{\theta'})) = \W{x}{\theta'} = W_1$.

Let $S = V \bigcap H^-(\B{x}{\theta'}, \y{x}{\theta'})$.
Since $W^*(x_U) < W_1 = W(S)$, $x_U$ is inside $CH(\mathcal{C}(S))$
by Lemma \ref{lem_point_roundch}.
Oppositely, by the definition of $S$,
$x$ is outside the convex hull $CH(\mathcal{C}(S))$.
It implies that $x_U$ is the topmost intersection point between
$CH(\mathcal{C}(S))$ and $L$,
hence on the upper boundary of $CH(\mathcal{C}(S))$.
(It is possible that $x_U$ locates at the
leftmost or the rightmost point of $CH(\mathcal{C}(S))$.)

The claimed angle $\theta$ is obtained as follows.
Since $x_U$ is a boundary point of $CH(\mathcal{C}(S))$,
there exists a line $F$ passing through $x_U$
and tangent to $CH(\mathcal{C}(S))$.
Let $\theta$ be the angle satisfying that
$\F{x_U}{\theta} = F$, $\pi \le \theta \le 2\pi$,
and $CH(\mathcal{C}(S)) \subset H(\F{x_U}{\theta}, \y{x_U}{\theta})$.
Obviously, we have that $S \subset H(\B{x_U}{\theta}, \y{x_U}{\theta})$
and thus $W(H(\B{x_U}{\theta}, \y{x_U}{\theta})) \ge W^*(S) = W_1$.
\qed
\end{proof}
}

Let $\theta_U$ be an arbitrary angle satisfying
the conditions of Lemma \ref{lem_xu_boundary}.
We apply the line $\F{x_U}{\theta_U}$
for trimming the region of $Wedge(x_U)$,
so that a sideward wedge can be obtained.
Let $PW(x_U)$, called the \emph{pseudo wedge} of $x_U$, denote the
intersection of $Wedge(x_U)$ and $H(\F{x_U}{\theta_U}, \y{x_U}{\theta_U})$.
Deriving from the three facts that $Wedge(x_U)$ is upward,
$\delta(Wedge(x_U)) < \pi$, and $\pi \le \theta_U \le 2\pi$,
we can observe that either $PW(x_U)$ is $x_U$ itself,
or it intersects only one of the right and left plane of $L$.
In the two circumstances, $PW(x_U)$ is said to be
\emph{null} or \emph{sideward}, respectively.
The pseudo wedge has similar functionality as wedges,
as shown in the following corollary.


\begin{cor} \label{cor_xu_PWprune}
    For any point $x' \notin PW(x_U)$, $W^*(x') \ge W^*(x_U)$.
\end{cor}

\Xomit{
\begin{proof}
If $x' \notin Wedge(x_U)$,
the lemma directly holds by Lemma \ref{lem_point_wedge}.
Otherwise, we have that $x' \notin H(\F{x_U}{\theta}, \y{x_U}{\theta})$
and thus $H(\B{x'}{\theta}, \y{x'}{\theta})$
contains $H(\B{x_U}{\theta}, \y{x_U}{\theta})$.
Then, by Lemma \ref{lem_xu_boundary}, 
$W^*(x') \ge W(H(\B{x'}{\theta}, \y{x'}{\theta}))
\ge W(H(\B{x_U}{\theta}, \y{x_U}{\theta})) \ge W_1$, 
thereby completes the proof.
\qed
\end{proof}
}

By this lemma, if $PW(x_U)$ is found to be sideward,
points on the opposite half-plane with respect to $L$ can be pruned.
If $PW(x_U)$ is null, $x_U$ becomes another kind of strong \CENR{1}{1}s,
in the meaning that it is also an
immediate solution to the \CENR{1}{1} problem.
Without confusion, we call $x_U$ a
\emph{conditional \CENR{1}{1}} in the latter case.

On the other hand, considering the reverse case that $W^*(x_D) < W^*(x_U)$,
we can also obtain an angle $x_D$ and a pseudo wedge $PW(x_D)$
for $x_D$ by symmetric arguments.
Then, either $PW(x_D)$ is sideward and the opposite side of $L$ can be pruned,
or $PW(x_D)$ itself is a conditional \CENR{1}{1}.
Thus, the third phase solves the problem
of the nonexistence of sideward wedges.

%
%

Recall that the three phases of searching sideward wedges
is based on the existence of $x_D$ and $x_U$ on $L$,
which was not guaranteed before.
Here we show that, by constructing appropriate border lines,
we can guarantee the existence of $x_D$ and $x_U$
while searching between these border lines.
The \emph{bounding box} is defined as the smallest axis-aligned rectangle
that encloses all circles in $\mathcal{C}(V)$.
Obviously, any point $x$ outside the box
satisfies that $W^*(x) = W(V)$ and must not be a \CENR{1}{1}.
Thus, given a vertical line not intersecting the box,
the half-plane to be pruned is trivially decided.
Moreover, let $T_\text{top}$ and $T_\text{btm}$ be two arbitrary
horizontal lines strictly above and below the bounding box, respectively.
We can obtain the following.

\begin{lem} \label{lem_bounding}
    Let $L$ be an arbitrary vertical line intersecting the bounding box,
    and $x'_D$ and $x'_U$ denote its intersection points
    with $T_\text{top}$ and $T_\text{btm}$, respectively.
    $Wedge(x'_D)$ is downward and $Wedge(x'_U)$ is upward.
\end{lem}

\Xomit{
\begin{proof}
Consider the case about $Wedge(x'_D)$.
As described above, we know the fact that $W^*(x_D) = W(V)$.
Let $\theta$ be an arbitrary angle with $0 \le \theta \le \pi$.
We can observe that $H^-(\F{x'_D}{\theta}, \y{x'_D}{\theta})$
cannot contain all circles in $\mathcal{C}(V)$,
that is, $V \not\subset H^-(\B{x'_D}{\theta}, \y{x'_D}{\theta})$.
This implies that $\W{x'_D}{\theta} < W^*(x'_D)$ and $\theta \notin MA(x'_D)$.
Therefore, we have that $MA(x'_D) \subset (\pi, 2\pi)$
and $Wedge(x'_D)$ is downward by Property \ref{pro_wedge_dir}.
By similar arguments, we can show that $Wedge(x'_U)$ is upward.
Thus, the lemma holds.
\qed
\end{proof}
}

According to this lemma,
by inserting $T_\text{top}$ and $T_\text{btm}$ into $\mathcal{T}$,
the existence of $x_D$ and $x_U$ is enforced
for any vertical line intersecting the bounding box.
Besides, it is obvious to see that the insertion does not affect
the correctness of all lemmas developed so far.

Summarizing the above discussion,
the whole picture of our desired pruning procedure can be described as follows.
In the beginning, we perform a preprocessing to obtain the bounding box
and then add $T_\text{top}$ and $T_\text{btm}$ into $\mathcal{T}$.
Now, given a vertical line $L$,
whether to prune its left or right plane
can be determined by the following steps.
\begin{enumerate}
    \item If $L$ does not intersect the bounding box,
        prune the half-plane not containing the box.
    \item Compute $x_D$ and $x_U$ on $L$.
    \item Find a sideward wedge or pseudo wedge via three forementioned phases.
        (Terminate whenever a strong or conditional \CENR{1}{1} is found.)
    \begin{enumerate}
        \item If breakpoints exist between $x_D$ and $x_U$,
            pick any of them and check it.
        \item If no such breakpoint,
            decide whether $L$ is non-leaning by checking $x_B$.
        \item If $L$ is non-leaning, compute $PW(x_U)$ or $PW(x_D)$
            depending on which of $x_U$ and $x_D$ has smaller weight loss.
    \end{enumerate}
    \item Prune the right or left plane of $L$ according to
        the direction of the sideward wedge or pseudo wedge.
\end{enumerate}

The correctness of this procedure follows from the developed lemmas.
Any vertical line not intersecting the bounding box
is trivially dealt with in Step 1, due to the property of the box.
When $L$ intersects the box, by Lemma \ref{lem_bounding},
$x_D$ and $x_U$ can certainly be found in Step 2.
The three sub-steps of Step 3 correspond to the three searching phases.
When $L$ is not non-leaning, a sideward wedge is found,
either at some breakpoint between $x_D$ and $x_U$ in Step 3(a)
by Corollary \ref{cor_breakpoint},
or at $x_B$ in Step 3(b) by Lemma \ref{lem_no_breakpoint}.
Otherwise, according to Lemma \ref{lem_xu_boundary} or its symmetric version,
a pseudo wedge can be built in Step 3(c) for $x_U$ or $x_D$, respectively.
Finally in Step 4, whether to prune the left or right plane of $L$
can be determined via the just-found sideward wedge or pseudo wedge,
by respectively Lemma \ref{lem_point_wedge} or Corollary \ref{cor_xu_PWprune}.


The time complexity of this procedure is analyzed as follows.
The preprocessing for computing the bounding box
trivially takes $O(n)$ time.
In Step 1, any vertical line not intersecting the box
can be identified and dealt with in $O(1)$ time.
Finding $x_D$ and $x_U$ in Step 2 requires the help
of the binary-search algorithm developed in \ref{ssec_LineLocal}.
Although the algorithm is designed to find a local optimal point,
we can easily observe that slightly modifying its objective
makes it applicable to this purpose without changing its time complexity.
Thus, Step 2 can be done in $O(n \log^2 n)$ time by Lemma \ref{lem_localopt}.

In Step 3(a), all breakpoints between $x_D$ and $x_U$
can be found in $O(n \log n)$ time as follows.
As done in Lemma \ref{lem_localopt},
we first list all breakpoints on $L$ by
$O(n)$ sorted sequences of length $O(n)$,
which takes $O(n \log n)$ time.
Then, by performing binary search with
the y-coordinates of $x_D$ and $x_U$,
we can find within each sequence the breakpoints
between them in $O(\log n)$ time.
In Step 3(a) or 3(b),
checking a picked point $x$ is done by computing $CA(x)$,
that requires $O(n \log n)$ time by Lemma \ref{lem_computeWedge}.
To compute the pseudo wedge in Step 3(c), the angle
$\theta_U$ satisfying Lemma \ref{lem_xu_boundary},
or symmetrically $\theta_D$, can be computed in $O(n \log n)$ time
by sweeping technique as in Lemma \ref{lem_computeWedge}.
Thus, $PW(x_U)$ or $PW(x_D)$ can be computed in $O(n \log n)$ time.
Finally, the pruning decision in Step 4 takes $O(1)$ time.
Summarizing the above, these steps require $O(n \log^2 n)$ time in total.
Since the invocation of Lemma \ref{lem_localopt}
needs an additional $O(n^2 \log n)$-time preprocessing,
we have the following result.

\begin{lem} \label{lem_linepruning}
    With an $O(n^2 \log n)$-time preprocessing,
    whether to prune the right or left plane of a given vertical line $L$
    can be determined in $O(n \log^2 n)$ time.
\end{lem}

\subsection{Searching on the Euclidean Plane} \label{ssec_PlaneSearch}

In this subsection, we come back to the \CENR{1}{1} problem.
Recall that, by Lemma \ref{lem_TTCCTC},
at least one \CENR{1}{1} can be found
in the three sets of intersection points
$\mathcal{T} \times \mathcal{T}$, $\mathcal{C}(V) \times \mathcal{T}$,
and $\mathcal{C}(V) \times \mathcal{C}(V)$,
which consist of total $O(n^4)$ points.
Let $\mathcal{L}$ denote the set of all vertical lines
passing through these $O(n^4)$ intersection points. 
By definition, there exists a vertical line $L^* \in \mathcal{L}$
such that its local optimal point is a \CENR{1}{1}.
Conceptually, with the help of Lemma \ref{lem_linepruning},
$L^*$ can be derived by applying prune-and-search approach to $\mathcal{L}$:
pick the vertical line $L$ from $\mathcal{L}$ with median x-coordinates,
determine by Lemma \ref{lem_linepruning}
whether the right or left plane of $L$ should be pruned,
discard lines of $\mathcal{L}$ in the pruned half-plane,
and repeat above until two vertical lines left.
Obviously, it costs too much if this approach is carried out
by explicitly generating and sorting the $O(n^4)$ lines.
However, by separately dealing with each of the three sets,
we can implicitly maintain sorted sequences of these lines
and apply the prune-and-search approach.

Let $\mathcal{L_T}$, $\mathcal{L_M}$, and $\mathcal{L_C}$
be the sets of all vertical lines passing through the intersection points
in $\mathcal{T} \times \mathcal{T}$, $\mathcal{C}(V) \times \mathcal{T}$,
and $\mathcal{C}(V) \times \mathcal{C}(V)$, respectively.
A \emph{local optimal line} of $\mathcal{L_T}$ is a vertical line $L_t^*$
such that its local optimal point has weight loss no larger than
those of points in $\mathcal{T} \times \mathcal{T}$.
The local optimal lines $L_m^*$ and $L_c^*$ can be similarly defined
for $\mathcal{L_M}$ and $\mathcal{L_C}$, respectively.
We will adopt different prune-and-search techniques
to find the local optimal lines in the three sets,
as shown in the following lemmas.

\begin{lem} \label{lem_TT}
    A local optimal line $L_t^*$ of $\mathcal{L_T}$
    can be found in $O(n^2 \log n)$ time.
\end{lem}

\Xomit{
\begin{proof}
Let $N_1 = |\mathcal{T}|$.
By definition, there are $(N_1)^2$ intersection points
in $\mathcal{T} \times \mathcal{T}$
and $(N_1)^2$ vertical lines in $\mathcal{L_T}$.
For efficiently searching within these vertical lines,
we apply the ingenious idea of parametric search
via parallel sorting algorithms,
proposed by Megiddo \cite{Megiddo-83-1}.

Consider two arbitrary lines $T_g, T_h \in \mathcal{T}$.
If they are not parallel, let $t_{gh}$ be their intersection point
and $L_{gh}$ be the vertical line passing through $t_{gh}$.
Suppose that $T_g$ is above $T_h$ in the left plane of $L_{gh}$.
If applying Lemma \ref{lem_linepruning} to $L_{gh}$ prunes its right plane,
$T_g$ is above $T_h$ in the remained left plane.
On the other hand, if the left plane of $L_{gh}$ is pruned,
$T_h$ is above $T_g$ in the remained right plane.
Therefore, $L_{gh}$ can be treated
as a ``\emph{comparison}'' between $T_g$ and $T_h$,
in the sense that applying Lemma \ref{lem_linepruning} to $L_{gh}$
determines their ordering in the remained half-plane.
It also decides the ordering of their intersection points
with the undetermined local optimal line $L_t^*$,
since the pruning ensures that
a local optimal line stays in the remained half-plane.

It follows that, by resolving comparisons,
the process of pruning vertical lines in $\mathcal{L_T}$ to find $L_t^*$
can be reduced to the problem of determining the ordering of
the intersection points of the $N_1$ lines with $L_t^*$,
or say, the sorting of these intersection points on $L_t^*$.
While resolving comparisons during the sorting process,
we can simultaneously maintain the remained half-plane
by two vertical lines as its boundaries.
Thus, after resolving all comparisons in $\mathcal{L_T}$,
one of the two boundaries must be a local optimal line. 
As we know, the most efficient way to obtain the ordering
is to apply some optimal sorting algorithm $A_S$,
which needs to resolve only $O(N_1 \log N_1)$ comparisons,
instead of $(N_1)^2$ comparisons.
Since resolving each comparison takes
$O(n \log^2 n)$ time by Lemma \ref{lem_linepruning},
the sorting is done in $O(n \log^2 n) \times O(N_1 \log^2 N_1)
= O(n^3 \log ^3 n)$ time, so is the finding of $L_t^*$.

However, Megiddo \cite{Megiddo-83-1} observed that,
when multiple comparisons can be indirectly resolved in a batch,
simulating parallel sorting algorithms in a sequential way
naturally provides the scheme for batching comparisons,
thereby outperform the case of applying $A_S$.
Let $A_P$ be an arbitrary cost-optimal parallel sorting algorithm
that runs in $O(\log n)$ steps on $O(n)$ processors,
e.g., the parallel merge sort in \cite{Cole-88}.
Using $A_P$ to sort the $N_1$ lines in $\mathcal{L_T}$
on $L_t^*$ takes $O(\log N_1)$ parallel steps.
At each parallel step, there are $k = O(N_1)$ comparisons
$L_1, L_2, \cdots, L_k$ to be resolved.
We select the one with median x-coordinate among them,
which is supposed to be some $L_i$.
If applying Lemma \ref{lem_linepruning} to $L_i$ prunes its left plane,
for each comparison $L_j$ to the left of $L_i$,
the ordering of the corresponding lines of $L_j$
in the remained right plane of $L_i$ is directly known.
Thus, the $O(k/2)$ comparisons to the left of $L_i$
are indirectly resolved in $O(k/2)$ time.
If otherwise the right plane of $L_i$ is pruned,
the $O(k/2)$ comparisons to its right are resolved in $O(k/2)$ time.
By repeating this process of selecting medians and pruning
on the remaining elements $O(\log k)$ times,
all $k$ comparisons can be resolved,
which takes $O(n \log^2 n) \times O(\log k) + O(k + k/2 + k/4 + \cdots)
= O(n \log^3 n) + O(N_1) = O(n^2)$ time.
Therefore, going through $O(\log N_1)$ parallel steps of $A_P$
requires $O(n^2 \log N_1) = O(n^2 \log n)$ time,
which determines the ordering of lines in $\mathcal{L_T}$ on $L_t^*$
and also computes a local optimal line $L_t^*$.
\qed
\end{proof}
}

\begin{lem} \label{lem_TC}
    A local optimal line $L_m^*$ of $\mathcal{L_M}$
    can be found in $O(n^2 \log n)$ time.
\end{lem}

\Xomit{
\begin{proof}
To deal with the set $\mathcal{L_M}$,
we use the ideas similar to the proofs of Lemmas
\ref{lem_buildsequence} and \ref{lem_localopt}
in order to divide $\mathcal{C}(V) \times \mathcal{T}$
into sorted sequences of points.
Given a fixed circle $C = \Cr{u_0}$ for some point $u_0 \in V$,
we show that the intersection points in $C \times \mathcal{T}^r(v)$
and $C \times \mathcal{T}^l(v)$ for each $v \in V$
can be grouped into $O(1)$ sequences of length $O(n)$,
which are sorted in increasing x-coordinates.
Summarizing over all circles in $\mathcal{C}(V)$,
there will be total $O(n^2)$ sequences of length $O(n)$,
each of which maps to a sequence of $O(n)$ vertical lines
sorted in increasing x-coordinates.
Then, finding a local optimal line $L_m^*$ can be done
by performing prune-and-search to the $O(n^2)$ sequences
of vertical lines via parallel binary searches.
The details of these steps are described as follows.

First we discuss about the way for grouping intersection points in
$C \times \mathcal{T}^r(v)$ and $C \times \mathcal{T}^l(v)$
for a fixed point $v \in V$, so that each of them
can be represented by $O(1)$ subsequences of $P(v)$.
By symmetry, only $C \times \mathcal{T}^r(v)$ is considered.
Similar to Lemma \ref{lem_buildsequence}, we are actually computing
sequences of points in $V \backslash \{v\}$
corresponding to these intersection points.
For each $v_i \in V \backslash \{v\}$,
the outer tangent line $\Tr{v}{v_i}$
may intersect $C$ at two, one, or zero point.
Let $\tr[C,1]{v}{v_i}$ and $\tr[C,2]{v}{v_i}$ denote
the first and second points, respectively,
at which $\Tr{v}{v_i}$ intersects $C$
along the direction from $v$ to $v_i$.
Note that, when $\Tr{v}{v_i}$ intersects $C$ at less than two points,
$\tr[C,2]{v}{v_i}$ or both of them will be \emph{null}.

In the following, we consider the sequence computation
under two cases about the relationship between $u_0$ and $v$,
(1) $u_0 = v$, and (2) $u_0 \neq v$.

\begin{description}
    \item [Case (1):]
        Since $v$ coincides with $u_0$,
        $C \times \mathcal{T}^r(v)$ is just the set of tangent points
        $\tr[C,1]{v}{v_i}$ for all $v_i \in V \backslash \{v\}$.
        It is easy to see the the angular sorted sequence $P(v)$
        directly corresponds to a sorted sequence of
        these $n-1$ tangent points in CCW order.
        $P(v)$ can be further partitioned into
        two sub-sequences $P_1(v)$ and $P_2(v)$,
        which consist of points in $V \backslash \{v\}$
        with polar angles (with respect to $v$)
        in the intervals $[0, \pi)$ and $[\pi, 2\pi)$, respectively.
        Since they are sorted in CCW order,
        we have that intersection points corresponding to
        $P_2(v)$ and to the reverse of $P_1(v)$
        are sorted in increasing x-coordinates, as we required.
        Obviously, $P_1(v)$ and $P_2(v)$ are of length $O(n)$
        and can be obtained in $O(\log n)$ time.

    \item [Case (2):]
        Suppose without loss of generality that
        $v$ locates on the lower left quadrant with respect to $u_0$,
        and let $\theta_0$ be the polar angle of $u_0$ with respect to $v$.
        This case can be further divided into two subcases
        by whether or not $\Cr{v}$ intersects $C$ 
        at less than two points.

        \begin{figure}[t]
            \centering
            \begin{minipage}[b]{0.58\textwidth}
                \centering
                \includegraphics[scale=0.75]{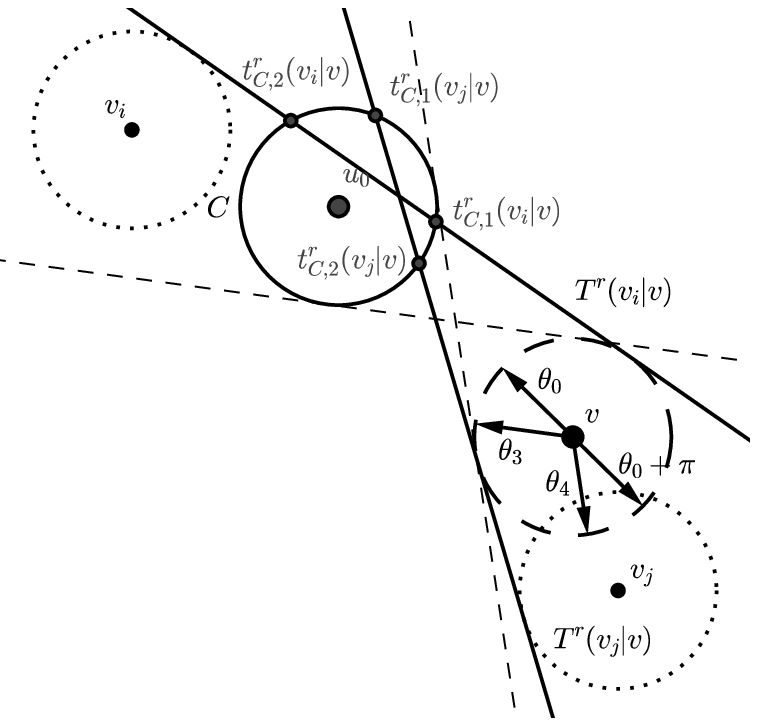}\\
                \scriptsize(a) no intersection
            \end{minipage}
            \hfill
            \begin{minipage}[b]{0.38\textwidth}
                \centering
                \includegraphics[scale=0.75]{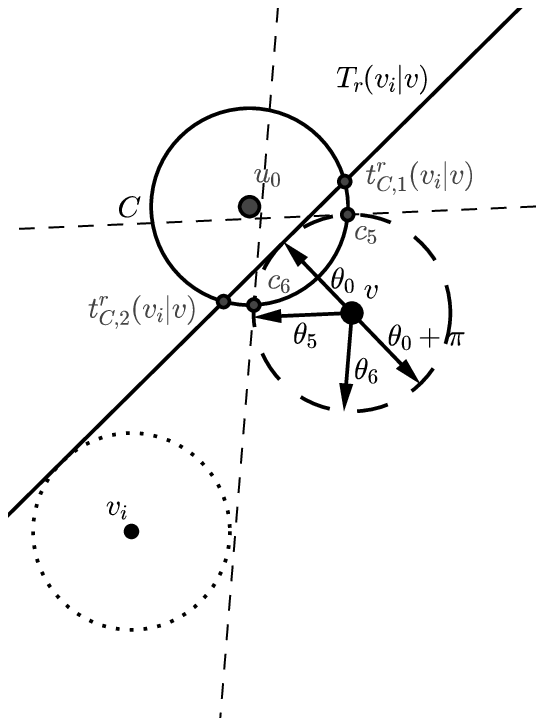}\\
                \scriptsize(b) two intersection points
            \end{minipage}
            \caption{Two subcases about how $\Cr{v}$ intersects $C$.}
            \label{fig_CCintersect}
        \end{figure}

        Consider first the subcase that they intersect
        at none or one point (see Figure \ref{fig_CCintersect}(a).)
        Let $\theta_3$ and $\theta_4$ be the angles
        such that $\Tr{v}{\y{v}{\theta_3}}$
        and $\Tr{v}{\y{v}{\theta_4}}$ are
        inner tangent to $\Cr{v}$ and $C$, 
        where $\theta_3 \le \theta_4$
        (Note that $\theta_3 = \theta_4$ only when
        the two circles intersect at one point.)
        For each $v_i \in V \backslash \{v\}$, $\Tr{v}{v_i}$
        does not intersect $C$,
        if the polar angle of $v_i$ with respect to $v$ is
        neither in $[\theta_0, \theta_3]$
        nor in $[\theta_4, \theta_0 + \pi]$.
        We can implicitly obtain from $P(v)$
        two subsequences $P_3(v)$ and $P_4(v)$,
        consisting of points with polar angles
        in $[\theta_0, \theta_3]$ and
        in $[\theta_4, \theta_0 + \pi]$, respectively.
        It can be observed that
        the sequence of points $v_i$ listed in $P_3(v)$
        corresponds to a sequence of intersection points
        $\tr[C,1]{v}{v_i}$ listed in clockwise (CW) order on $C$
        and, moreover, a sequence of $\tr[C,2]{v}{v_i}$
        listed in CCW order on $C$.
        Symmetrically, the sequence of points $v_j$ in $P_4(v)$
        corresponds to a sequence of $\tr[C,1]{v}{v_j}$ in CCW order
        and a sequence of $\tr[C,2]{v}{v_j}$ in CW order.
        The four implicit sequences of intersection points on $C$
        can be further partitioned by a horizontal line
        $L_h$ passing through its center $u_0$,
        so that the resulted sequences are naturally sorted
        in either increasing or decreasing x-coordinates.
        Therefore, we can implicitly obtain
        at most eight sorted sequences of length $O(n)$
        in replace of $C \times \mathcal{T}^r(v)$,
        by appropriately partitioning $P(v)$ in $O(\log n)$ time.

        Consider that $\Cr{v}$ intersects $\Cr{u_0}$
        at two points $c_5$ and $c_6$,
        where $c_5$ is to the upper right of $c_6$
        (see Figure \ref{fig_CCintersect}(b).)
        Let $\theta_5$ and $\theta_6$ be the angles such that
        $\Tr{v}{\y{v}{\theta_5}}$ and $\Tr{v}{\y{v}{\theta_6}}$
        are tangent to $\Cr{v}$ at $c_5$ and $c_6$, respectively.
        Again, $P(v)$ can be implicitly partitioned into three subsequences
        $P_5(v)$, $P_6(v)$, and $P_7(v)$, which consists of points
        with polar angles in $[\theta_0, \theta_5)$, $[\theta_5, \theta_6)$,
        and $[\theta_6, \theta_0 + \pi]$, respectively.
        By similar observations, $P_5(v)$ corresponds to
        two sequences of intersection points
        listed in CW and CCW order, respectively,
        and $P_7(v)$ corresponds to two sequences
        listed in CCW and CW order, respectively.
        However, the sequence of points $v_i$ in $P_6(v)$ corresponds to
        the sequences of $\tr[C,1]{v}{v_i}$ and
        $\tr[C,2]{v}{v_i}$ listed in both CCW order.
        These sequences can also be partitioned by $L_h$
        into sequences sorted in x-coordinates.
        It follows that we can implicitly obtain
        at most twelve sorted sequences of length $O(n)$
        in replace of $C \times \mathcal{T}^r(v)$ in $O(\log n)$ time.
\end{description}

According to the above discussion, for any two points $u, v \in V$,
$\Cr{u} \times \mathcal{T}^r(v)$ and $\Cr{u} \times \mathcal{T}^l(v)$
can be divided into $O(1)$ sequences in $O(\log n)$ time,
each of which consists of $O(n)$ intersection points on $\Cr{u}$
sorted in increasing x-coordinates.
Thus, $\mathcal{C}(V) \times \mathcal{T}$ can be re-organized as
$O(n^2)$ sorted sequences of length $O(n)$ in $O(n^2 \log n)$ time,
which correspond to $O(n^2)$ sorted sequences of $O(n)$ vertical lines.
Now, we can perform parametric search for parallel binary search
to these sequences of vertical lines,
by similar techniques used in Lemma \ref{lem_localopt}.
For each of the $O(n^2)$ sequences, its middle element is first obtained
and assigned with a weight equal to the sequence length in $O(1)$ time.
Then, the weighted median $L$ of these $O(n^2)$ elements are
computed in $O(n^2)$ time \cite{Reiser-78}.
By applying Lemma \ref{lem_linepruning} to $L$ in $O(n \log^2 n)$ time,
at least one-eighths of total elements can be pruned from these sequences,
taking another $O(n^2)$ time.
Therefore, a single iteration of pruning requires $O(n^2)$ time.
After $O(\log n)$ such iterations,
a local optimal line $L_m^*$ can be found in total $O(n^2 \log n)$ time,
thereby proves the lemma.
\qed
\end{proof}
}

\begin{lem} \label{lem_CC}
    A local optimal line $L_c^*$ of $\mathcal{L_C}$
    can be found in $O(n^2 \log n)$ time.
\end{lem}

\Xomit{
\begin{proof}
There are at most $O(n^2)$ points
in $\mathcal{C}(V) \times \mathcal{C}(V)$.
Thus, $\mathcal{L_C}$ can be obtained
and sorted according to x-coordinates in $O(n^2 \log n)$ time.
Then, by simply performing binary search
with Lemma \ref{lem_linepruning},
a local optimal line $L_c^*$ can be easily found
in $O(\log n)$ iterations of pruning,
which require total $O(n \log^3 n)$ time.
In summary, the computation takes $O(n^2 \log n)$ time,
and the lemma holds.
\qed
\end{proof}
}

By definition, $L^*$ can be found among $L_t^*$, $L_m^*$, and $L_c^*$,
which can be computed in $O(n^2 \log n)$ time by Lemmas
\ref{lem_TT}, \ref{lem_TC}, and \ref{lem_CC}, respectively.
Then, a \CENR{1}{1} can be computed as the local optimal point
of $L^*$ in $O(n \log^2 n)$ time by Lemma \ref{lem_localopt}.
Combining with the $O(n^2 \log n)$-time preprocessing
for computing the angular sorted sequence $P(v)$s 
and the bounding box enclosing $\mathcal{C}(V)$,
we have the following theorem.

\begin{thm}
    The \CENR{1}{1} problem can be solved in $O(n^2 \log n)$ time.
\end{thm}

\section{Concluding Remarks} \label{conclusion}

In this paper, we revisited the \CEN{1}{1} problem on the Euclidean plane
under the consideration of minimal distance constraint between facilities,
and proposed an $O(n^2 \log n)$-time algorithm, 
which close the bound gap between this problem and its unconstrained version.
Starting from a critical observation on the medianoid solutions,
we developed a pruning tool with indefinite region remained after pruning,
and made use of it via multi-level structured parametric search approach,
which is quite different to the previous approach in \cite{Drezner-82,Hakimi-90}.

Considering distance constraint between facilities 
in various competitive facility location models
is both of theoretical interest and of practical importance.
However, similar constraints are rarely seen in the literature.
It would be good starting points by introducing the constraint
to the facilities between players in the \MED{r}{X_p} and \CEN{r}{p} problems,
maybe even to the facilities between the same player.




\end{document}